\documentclass{sl}     

\usepackage{slsec}     
\usepackage{stud_log} 
\usepackage{slfoot}
\usepackage{slthm}     
                         
\usepackage{graphicx,xcolor}

\newtheorem{theorem}{Theorem}[section]
\newtheorem{lemma}[theorem]{Lemma}

\newtheorem{corollary}[theorem]{Corollary}

\newtheorem{definition}[theorem]{Definition}
\newtheorem{example}[theorem]{Example}

\newcommand{\allpub}{\ensuremath{[ \bullet ]}}
\newcommand{\somepub}{\ensuremath{\langle \bullet \rangle}}
\newcommand{\allgrp}[1]{\ensuremath{[ #1 ]}}
\newcommand{\somegrp}[1]{\ensuremath{\langle #1 \rangle}}
\newcommand{\allcoal}[1]{\ensuremath{[\!\langle #1 \rangle\!]}}
\newcommand{\somecoal}[1]{\ensuremath{\langle\![ #1 ]\!\rangle}}

\newcommand{\know}{\ensuremath{K}}
\newcommand{\susp}{\ensuremath{\widehat{K}}}
\newcommand{\et}{\ensuremath{\land}}
\newcommand{\imp}{\ensuremath{\rightarrow}}
\newcommand{\Agents}{\ensuremath{A}}
\newcommand{\lang}{\ensuremath{\mathcal{L}}}
\newcommand{\States}{\ensuremath{S}}
\newcommand{\Atoms}{\ensuremath{P}}

\usepackage{amsfonts, amsmath, amssymb}
\renewcommand{\square}{\Box}

\usepackage{tikz}

\newcommand{\vel}{\vee}
\newcommand{\atom}{p}
\newcommand{\atomb}{q}
\newcommand{\agent}{a}

\newcommand{\agentb}{b}

\newcommand{\N}{\ensuremath{\mathbb{N}}}
\newcommand{\Group}{G}

\HeadingsInfo{{\AA}gotnes, van Ditmarsch, French}{Undecidability of Quantified Announcements}

\begin{document}

\newcommand{\tedge}{\ensuremath{\mathfrak{e}}}
\newcommand{\tsq}{\ensuremath{\mathfrak{s}}}
\newcommand{\tup}{\ensuremath{u}}
\newcommand{\tdown}{\ensuremath{d}}
\newcommand{\tlef}{\ensuremath{\ell}}
\newcommand{\trigh}{\ensuremath{r}}
\newcommand{\tcentr}{\ensuremath{c}}

\threeAuthorsTitle{T.\ts {\AA}gotnes}{H.\ts van Ditmarsch}{T.\ts French}{The Undecidability of Quantified Announcements}


\begin{abstract}
This paper demonstrates the undecidability of a number of logics with quantification over public announcements: arbitrary public announcement logic (APAL), group announcement logic (GAL), and coalition announcement logic (CAL). In APAL we consider the informative consequences of {\em any} announcement, in GAL we consider the informative consequences of a group of agents (this group may be a proper subset of the set of all agents) all of which are simultaneously (and publicly) making known announcements. So this is more restrictive than APAL. Finally, CAL is as GAL except that we now quantify over anything the agents not in that group may announce simultaneously as well. The logic CAL therefore has some features of game logic and of ATL. We show that when there are multiple agents in the language, the satisfiability problem is undecidable for APAL, GAL, and CAL. 
In the single agent case, the satisfiability problem is decidable for all three logics.

This paper corrects an error to the submitted version of 
{\em Undecidability of Quantified Announcements}, \cite{original} 
identified by Yuta Asami \cite{asami}. 
The nature of the error was in the definition of the formula $c_{ga}(X)$ (see Subsection~\ref{subsec.group}) which is corrected in this version.
\end{abstract}


\section{Introduction}

The runs-and-systems approach \cite{faginetal:1995} and, later, dynamic epistemic logics \cite{hvdetal.del:2007} have been used to specify multi-agent systems dynamic at a high level of systems architecture. Such modal logical approaches have some advantages, e.g.\ that they are typically decidable, the model checking complexity is fairly low (even linear for the base modal logic), and satisfiability may also be tamed quite a bit (public announcement logic and the base modal logic are NP-complete in the single-agent case and PSPACE-complete in the multi-agent case \cite{lutz:2006}). Unfortunately, when combining epistemic (or more basic modal) logics with logical dynamics, it is highly unpredictable if the resulting logic is decidable even when the base logics are so. A well-known result achieved by very minimal (linguistic) means is \cite{milleretal:2005}. Other such results, described below, were reported in \cite{frenchetal:2008,agotnesetal:2014}.
There are several open problems concerning decidability of dynamic epistemic logics and in this paper we answer some such questions negatively: the related logics arbitrary public announcement logic, group announcement logic, and coalition announcement logic are all undecidable. 
For the development of proof tools for such logics we consider this result relevant to report. 
Given the frequent claim that such logics are applicable, for example, for epistemic protocol synthesis \cite{bolanderetal:2011,bolanderetal:2015}, such undecidability results are disappointing. 
We can keep in mind there that model checking is still decidable and in the final section we discuss some ways towards decidable versions of these logics.

\paragraph*{APAL}
Public Announcement Logic (PAL) \cite{plaza:1989} is an extension of epistemic logic with dynamic modal operators for the effect of so-called {\em truthful public announcements}. The word `announcement' is a bit of a misnomer, nothing really needs to be said: in fact these announcements are reliable public observations, supposedly coming from outside of the modelled system and incorporated by all agents on the assumption of their correctness. Given a (Kripke) model consisting of epistemic alternatives, the execution of a public announcement is the restriction of the model to the submodel satisfying the announcement formula: $\langle\phi\rangle\psi$ means that the announcement formula $\phi$ is true and that in the restriction of the model to the $\phi$ states, $\psi$ is true (the part $\langle\phi\rangle$ is the dynamic modality, in its diamond form). It is well-known that this logic is equally expressive as epistemic logic: by a reduction technique any formula is equivalent to one without announcements. In Arbitrary Public Announcement Logic (APAL) \cite{balbianietal:2008} we quantify over such announcements, with the restriction that we only quantify over quantifier-free announcement formulas. The expression $\langle\bullet\rangle\phi$ means that for all $\langle\bullet\rangle$-free formulas $\psi$, $\langle\psi\rangle\phi$ is true, in other words: ``there is an announcement $\psi$ such that $\psi$ is true and after which $\phi$ is true.'' Unlike PAL, this logic APAL is more expressive than epistemic logic, and it is also undecidable \cite{frenchetal:2008}.  

\paragraph*{GAL}
Group Announcement Logic (GAL) \cite{agotnesetal:2008,agotnesetal.jal:2010} is an extension of public announcement logic that is very similar to APAL. The difference is that we now do not quantify over all epistemic formulas but over epistemic formulas known by the individual agents in a given group, that may be a subgroup of the set of all agents. In order to see the rationale of this, we need to step back a bit to the conditions of truthful public announcements. The `truthful' in PAL actually means `true'. It is common to model the truthful (true) public announcement of a formula $\phi$ made by an agent $\agent$ that is modelled in the system, as the announcement by an outsider of the formula $\know_\agent \phi$, for ``the agent $\agent$ knows $\phi$.'' Now we have `truthful' in its common sense, because agent $\agent$ truthfully saying that $\phi$ implies the truth of $\know_\agent \phi$. In group announcement logic we investigate what can be achieved by simultaneous truthful announcements by a subset of the set of all agents: $\langle\Group\rangle\phi$ means that the agents in $\Group$ can simultaneously say something they know, after which $\phi$ is true. Here, `simultaneously' merely means `without interference'. We can for example envisage this as the agents putting their known information in an envelope, after which the envelopes are opened in public. The validity $\langle\Group\rangle\langle\Group\rangle\phi\imp\langle\Group\rangle\phi$ of GAL says that we can always make two subsequent group announcements into a single one. Because of this property, such group announcements also quantify over finite (public) communication protocols between the agents in a group and where agents take turns in saying something: whenever some $\agent\in\Group$ is executing a step in the protocol, we think of all other agents $\agentb$ as performing the trivial announcement; they say $\know_\agentb \top$ (`I know that true is true'). Therefore, each step in such a protocol can be seen as a $\Group$-announcement, after which we apply the property above to make a single $\Group$-announcement from any finite sequence of such steps. The logic GAL shares various properties with APAL, e.g., the axiomatization is similar, and the model checking complexity is PSPACE-complete \cite{agotnesetal.jal:2010}. In \cite{agotnesetal:2014} the logic GAL has shown to be also undecidable. 

The expression $\somegrp{G} \phi$ has the smell of `group of agents $G$ is able to achieve $\phi$', such that, taking a single agent, $\know_\agent \somegrp{\{\agent\}} \phi$ seems to formalize that agent $\agent$ knows that she is able to achieve $\phi$, as in logics combining agency and knowledge \cite{agotnes:synthese-06,vdhoeketal:2002c}. The problem is that in different states, different announcements may be required to make $\phi$ true, so that we can have $\langle \know_\agent\psi\rangle \phi$ in one state and $\langle \know_\agent\chi\rangle \phi$ in an indistinguishable state, for non-equivalent (and non-dependent) $\know_\agent\psi$ and $\know_\agent\chi$: the agent knows she `can make' $\phi$ true, but still cannot choose between announcing $\psi$ or announcing $\chi$. These matters are also reported on in \cite{agotnesetal.jal:2010}.

\paragraph*{CAL}
A variation on GAL is coalition announcement logic (CAL) \cite{agotnesetal:2008}. In group announcement logic we investigate the consequences of the simultaneous announcement (joint public event) by $G$. The agents not in $G$ do not take part in the action. In CAL, the expression $\somecoal{G}\phi$ means that the agents in $G$ can simultaneously announce something they know, after which $\phi$ is true, no matter what the remaining agents also simultaneously announced. Clearly, we get into the domain of game logic \cite{Pauly2002} and ATL \cite{Alur2002}, see also the recent survey \cite{jfak.logicingames:2014}; and in fact CAL subsumes game logic (for example the axiom $\neg\somecoal{\emptyset}\neg\phi \imp \somecoal{\Agents}\phi$, where $\Agents$ is the set of all agents) \cite{agotnesetal:2008}. Unlike APAL and GAL, the logic CAL has no known axiomatization. Up to now it was not known that it is also undecidable. The $\somecoal{G}\phi$ operator further comes close to a {\em playability} operator as in \cite{jfak.tark:2001,jfak.logicingames:2014} (in a setting of information changing games, not for games with factual change).

\paragraph*{Example}
Before we present some other relevant properties of the presented logics, let us first give an example. Given are two agents Anne and Bill ($\agent$ and $\agentb$) such that Anne knows whether $\atom$ and Bill knows whether $\atomb$, and this is common knowledge (Figure~\ref{a-announce}, leftmost square). The names of the states are suggestive of the valuation of the atoms, e.g.\ in state $\overline{p}q$ we have that $p$ is false and $q$ is true. Let us call the model $\mathcal{S}$. In the figure we depict a number of executions of announcements in the form of announcing the {\em value} of some proposition $\phi$, i.e., an announcement of $\phi$ whenever it is true, and otherwise the announcement of $\neg\phi$. In the explanations we already use the semantic notation that will be introduced in the next section, i.e., $\mathcal{S}_s \models \phi$ means that $\phi$ is true in state (world) $s$ of model $\mathcal{S}$; $\mathcal{S} \models \phi$ means that $\phi$ is valid on model $\mathcal{S}$; and $\mathcal{S}^\phi$ is the restriction (submodel) of $\mathcal{S}$ to the states where $\phi$ is true.

{\bf APAL} In APAL, from any state in this model, any model restriction is definable by an epistemic formula. We have depicted the result of announcing the value of $p \vel q$, i.e., the announcement is $p \vel q$ in the states $pq$, $\overline{p}q$, and $p\overline{q}$, whereas the announcement is $\neg (p \vel q)$ in the state $\overline{p}\overline{q}$. So, for example, $\mathcal{S}_{\overline{p}q} \models \neg \know_a q \et \langle\bullet\rangle \know_a q$, because $\mathcal{S}_{\overline{p}q} \models \neg \know_a q$ and $\mathcal{S}_{\overline{p}q} \models \langle p \vel q \rangle \know_a q$; where the latter is true because $\mathcal{S}^{p\vel q}_{\overline{p}q} \models \know_a q$.

{\bf GAL}
Anne can achieve that Bill knows whether $\atom$, namely by informing him of the value of $\atom$ (Figure~\ref{b-announce}, middle transition). In other words, there is an announcement that $a$ can make and that $a$ knows to be true (namely $\know_a p$ when $p$ is true, and $\know_a \neg p$ when $p$ is false) such that after that announcement, $b$ knows whether $p$. We have that $\mathcal{S}_{pq} \models \langle a \rangle \know_b p$ because $\mathcal{S}_{pq} \models \langle \know_a p \rangle \know_b p$. The formalization of ``Anne can achieve that Bill knows whether $\atom$'' is  $\mathcal{S} \models \somegrp{a} (\know_b \atom \vel \know_b \neg \atom)$. Similarly, Bill can achieve that Anne knows whether $\atomb$, in a formula, $\mathcal{S} \models \somegrp{b} (\know_a \atomb \vel \know_a \neg \atomb)$ (Figure~\ref{b-announce}, first bottom transition). Neither agent can achieve both outcomes at the same time, but together they can achieve that by a joint announcement of what they know: $\somegrp{{ab}} (\know_a (\atom\et\atomb) \et \know_b (\atom\et\atomb))$ (Figure~\ref{b-announce}, bottom rightmost model). In the figure the announcement is not jointly by Anne and Bill, but still we can get there. This is because from $\somegrp{b} (\know_a \atomb \vel \know_a \neg \atomb)$ we also get $\somegrp{{ab}} (\know_a \atomb \vel \know_a \neg \atomb)$, by having Anne make the trivial announcement simultaneously to Bill announcing the value of $\atomb$. Similarly, from $\somegrp{a} (\know_b \atom \vel \know_b \neg \atom)$ we also get $\somegrp{{ab}} (\know_b \atom \vel \know_b \neg \atom)$. Now Bill is making the trivial announcement. We thus get $\somegrp{{ab}}\somegrp{{ab}} ((\know_b \atom \vel \know_b \neg \atom) \et (\know_a \atomb \vel \know_a \neg \atomb))$, and as such operators satisfy transitivity (see above) we get the desired $\somegrp{{ab}} ((\know_b \atom \vel \know_b \neg \atom) \et (\know_a \atomb \vel \know_a \neg \atomb))$. Of course, in the case of this simple example it suffices to observe that Anne can announce the value of $\atom$ simultaneously with Bill announcing the value of $\atomb$.

{\bf CAL} One can easily see that Anne can enforce knowledge of the value of $p$ on Bill, $\somecoal{a} (\know_b \atom \vel \know_b \neg \atom)$, namely by announcing the value of $p$, as before. This is true when Bill makes the trivial announcement but also when he announces the value of $q$, and there is nothing else he can do. But Anne cannot enforce her continued ignorance of $q$ so we have $\mathcal{S} \not\models \somecoal{a} \neg(\know_a \atomb \vel \know_a \neg \atomb)$. (Unlike that, we have $\mathcal{S} \models \somegrp{a} \neg(\know_a \atomb \vel \know_a \neg \atomb)$.) Together, they can enforce shared knowledge of the atoms: $\somecoal{{ab}} (\know_a (\atom\et\atomb) \et \know_b (\atom\et\atomb))$. This is obvious, as the entire group of agents is up against nobody, so $\mathcal{S} \models\somegrp{ab} \phi \leftrightarrow \somecoal{ab} \phi$ for all $\phi$.

\begin{figure}
\noindent
\scalebox{.9}{
\begin{tikzpicture}
\node (00a) at (0,1){\fbox{
\begin{tikzpicture}[z=0.35cm]
\node (00) at (0,0) {$\overline{p}\overline{q}$};
\node (01) at (0,2) {$\overline{p}q$};
\node (10) at (2,0) {$p\overline{q}$};
\node (11) at (2,2) {$pq$};
\draw (00) -- node[above] {$b$} (10);
\draw (00) -- node[left] {$a$} (01);
\draw (10) -- node[left] {$a$} (11);
\draw (01) -- node[above] {$b$} (11);
\end{tikzpicture}
}}
;
\node (11a) at (9,3){\fbox{
\begin{tikzpicture}[z=0.35cm]
\node (00) at (0,0) {$\overline{p}\overline{q}$};
\node (01) at (0,2) {$\overline{p}q$};
\node (10) at (2,0) {$p\overline{q}$};
\node (11) at (2,2) {$pq$};
%
\draw (10) -- node[left] {$a$} (11);
\draw (01) -- node[above] {$b$} (11);
\end{tikzpicture}
}}
;
\draw[->] (00a) -- node[above] {$\text{Outsider}$} (11a)
;
\node (10a) at (9,-1){\fbox{
\begin{tikzpicture}[z=0.35cm]
\node (00) at (0,0) {$\overline{p}\overline{q}$};
\node (01) at (0,2) {$\overline{p}q$};
\node (10) at (2,0) {$p\overline{q}$};
\node (11) at (2,2) {$pq$};
\draw (00) -- node[left] {$a$} (01);
\draw (10) -- node[left] {$a$} (11);
\end{tikzpicture}
}}
;
\draw[->] (00a) -- node[above] {$\text{Anne}$} (10a)
;
\node (1m1a) at (0,-3){\fbox{
\begin{tikzpicture}[z=0.35cm]
\node (00) at (0,0) {$\overline{p}\overline{q}$};
\node (01) at (0,2) {$\overline{p}q$};
\node (10) at (2,0) {$p\overline{q}$};
\node (11) at (2,2) {$pq$};
\draw (00) -- node[above] {$b$} (10);
\draw (01) -- node[above] {$b$} (11);
\end{tikzpicture}
}}
;
\draw[->] (00a) -- node[left] {$\text{Bill}$} (1m1a)
;
\node (2m1a) at (5,-3){\fbox{
\begin{tikzpicture}[z=0.35cm]
\node (00) at (0,0) {$\overline{p}\overline{q}$};
\node (01) at (0,2) {$\overline{p}q$};
\node (10) at (2,0) {$p\overline{q}$};
\node (11) at (2,2) {$pq$};
\end{tikzpicture}
}}
;
\draw[->] (1m1a) -- node[above] {$\text{Anne}$} (2m1a);
\end{tikzpicture}
}
\caption{Top transition: An announcement of the value of $\atom\vel\atomb$ that cannot be made by Anne or Bill. Middle transition: Anne announces the value of $\atom$. Bottom transitions: Bill announces the value of $\atomb$, after which Anne announces the value of $\atom$\label{a-announce}\label{b-announce}\label{out-announce}.}
\end{figure}
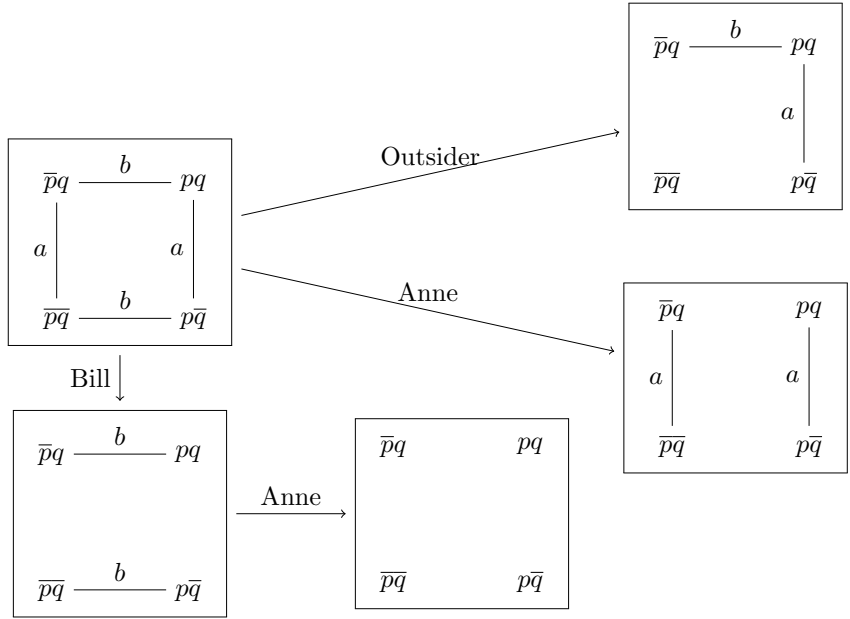

\paragraph*{Relative expressivity of APAL, GAL, and CAL} 
Single-agent APAL is equally expressive as PAL (and as the base epistemic logic) \cite{balbianietal:2008}, and this result also carries over to GAL and CAL, so for this expressivity survey we assume that there is more than one agent. 

It will be clear that a group announcement is an announcement, so that $\somegrp{G}\phi \imp \langle\bullet\rangle\phi$; but that, as in the above example, there are announcements that not group announcements. This can be used to demonstrate that APAL is more expressive than GAL. Whether GAL is also more expressive than APAL is an open question. It is also unknown what the relative expressivity is of CAL and GAL. Such matters are discussed in \cite{agotnesetal.jal:2010,hvd.wollic:2012}. For example, a GAL validity is $\somegrp{G}\somegrp{H}\phi\imp\somegrp{G\cup H}\phi$ but it is unknown if $\somegrp{G\cup H}\phi\imp\somegrp{G}\somegrp{H}\phi$ is valid in CAL: this involves constructing a {\em later} announcement by the agents in $H$ with the same informative effect as their immediate announcement. On models with infinite domains it is unclear how this should be done (and even if it can be done at all); note that the agents in $G$ can efface distinctions between previously bisimilar states that at the time could have been used by those in $H$, but not any more after $G$'s announcement. Similar issues complicate proof of the conjecture that GAL and CAL are equally expressive: we wish to compare the expression $\somecoal{G}\phi$ to the expression $\somegrp{G}\allgrp{\Agents\setminus{G}}\phi$ (where $\Agents$ is the set of all agents): one would expect that quantifying over anything the remaining agents can say {\em at the same time} as those in $G$ should be the same as quantifying over what they can say {\em afterwards}. But to our knowledge this is yet another open question. Answering such questions is also relevant for yet other logics with quantification. For example, consider the single-agent fragment of GAL where we only have operators $\somegrp{\agent}$, for any $\agent\in\Agents$.\footnote{This interesting fragment was suggested by a reviewer, and independently by Jan van Eijck in oral communication.} Clearly, in view of the above, $\somegrp{\{\agent,\agentb\}}\phi$ need not be equivalent to $\somegrp{\agent}\somegrp{\agentb}\phi$. 

\paragraph*{GAL and security}
With GAL one can formalize communication protocols, such as security protocols. Let Alice be a sender $a$, Bob a receiver $b$, and Eve a spy / eavesdropper $e$. Let $\phi$ be some {\em information goal}. For example, suppose Alice wishes to inform Bob of the latest transatlantic scandal $\atom$, then an information goal could be that ``Alice, Bob, and Eve commonly know that either Alice and Bob have common knowledge (between them) of $\atom$ or Alice and Bob have common knowledge of $\neg\atom$.'' This can be more succinctly formalized with common knowledge operators but we bypass these in this paper. The requirement that not only Alice and Bob but also they and Eve commonly know this, is usual in a security setting. It formalizes that the protocol is known to have terminated: we may assume that everything is public about the protocol except the message (and private keys). There is also a {\em security goal} $\psi$ that needs to be preserved throughout protocol execution, e.g.\ ``Alice, Bob, and Eve commonly know that Eve is ignorant about $\atom$'' (or some more involved aspect of $\atom$, such as the identity of those involved in the scandal). A finite protocol for $a$ and $b$ to learn the secret safely should observe \[ \psi \imp \somegrp{{ab}} (\phi\et\psi) \]
In view of the open expressivity problems mentioned above, it is not known if the existence of a finite two-agent protocol specification as above is formalizable in APAL.

Such formalizations of security protocols in GAL may be relevant to obtain results on communication complexity. It is well-know that for $n$ agents to share their knowledge, $n$ rounds of communication are sufficient \cite{jfak.lonely:2006}. Still, there are many security protocols between a sender and a receiver that require more than two communications (announcements), even in the above very restricted setting. An example are Russian Cards protocols of three or more announcements, that are proved not be possible in two announcements, as described in \cite{cordonetal.tcs:2013}. The above $\psi \imp \somegrp{{ab}} (\phi\et\psi)$ can be realized by a single announcement by $a$ and a single announcement by $b$ (and even by a single simultaneous announcement by $a$ and $b$), by the procedure of merging successive announcements into a single announcement. However, even though we may have $\somegrp{{ab}} (\phi\et\psi)$, or, in a more restricted protocol setting, that there is a protocol execution by $a$ and $b$ after which $\phi\et\psi$ is true, that does not imply that {\em agent $\agent$ knows that}, or that {\em $\agentb$ knows that}. Typically, in a security setting, an agent will only make an announcement if she {\em knows} that the eavesdropper remains ignorant afterwards. Merely considering this possible is not enough. This explains the need for protocols consisting of more than two steps. As such, GAL or CAL formalizations are rather epistemic protocol specifications (with the group and coalition operators performing the role of temporal modalities in the better known temporal epistemic specifications) that we then wish to synthesize from such descriptions. In other words, such GAL and CAL specifications (and, similarly APAL specifications) can be a role in the already mentioned epistemic planning \cite{bolanderetal:2011}.

\paragraph*{Results in the paper}

APAL was shown to be undecidable in \cite{frenchetal:2008}. GAL was shown to be undecidable in \cite{agotnesetal:2014}. The satisfiability problem for APAL and GAL is not recursively enumerable. Also, based on the complete axiomatization of APAL, so that APAL validity is recursively enumerable, we may conclude that satisfiability for APAL is co-RE complete, and similarly for GAL. We further show that CAL is also undecidable. This result is new in this paper. The undecidability proofs for APAL, GAL, and CAL are all sufficiently different to warrant a integrated treatment that points out their different proofs in detail, and the presentation of the undecidability proofs of APAL and GAL in this paper is therefore substantially different from those in \cite{frenchetal:2008,agotnesetal:2014}. Whilst the previous undecidability results used a construction that required five agents, here we present an improved construction that only requires two agents. As the satisfiability for the single agent fragment of these logics is trivially decidable, this gives us a complete picture of decidability for these logics.

\paragraph{Overview of the paper}
Section \ref{sec.syntaxsemantics} introduces the logics APAL, GAL, and CAL. Sections~\ref{sec.tilings}~and~\ref{sec.bisim}  present the background on how to prove undecidability from a tiling argument. The core of the paper is then Section \ref{sec.grid} wherein the grid-like structure is defined in order to enforce a tiling argument, for the three respective logics APAL (Subsection \ref{subsec.arbitrary}), GAL (Subsection \ref{subsec.group}), and CAL (Subsection \ref{subsec.coalition}). Based on this groundwork, Section \ref{sec.together} then establishes the undecidability results in a uniform way for the three respective logics. The final Section \ref{sec.future} discusses possible generalizations and restrictions, namely to non-public communications, positive announcements, and group epistemic operators.

\section{Syntax and semantics} \label{sec.syntaxsemantics}

Given is a scenario with a finite group of agents $\Agents$ and a countable set of atomic (Boolean) propositions $\Atoms$. The agents consider different worlds possible where different sets of propositions may be true, and where different agents may have different states of knowledge. Our interest is in providing a formal language for reasoning about what agents know of the propositions, what agents know about what other agents know, and what agents can find out through informative events (such as public announcements).  

The base language we work with is epistemic logic EL:\\
$$ \phi ::=\ p\ |\ \lnot\phi\ |\ (\phi\land\phi)\ |\ \know_a\phi$$
where $p\in\Atoms$ and $a\in\Agents$. As usual, we take \(\know_a\phi\) to mean {\em agent $a$ knows }\(\phi\), and  let \(\susp_a\phi\) abbreviate \(\lnot \know_a\lnot\phi\) ({\em agent $a$ considers }$\phi$ {\em to be possible}). We consider extensions of this logic with
\begin{enumerate}
\item public announcements $[\psi]\phi$;
\item arbitrary public announcements $\allpub\phi$;
\item group announcements $\allgrp{G}\phi$; 
\item coalition announcements $\allcoal{G}\phi$;
\end{enumerate}
where $G\subseteq\Agents$ and $\psi$ and $\phi$ may be any formula of the (extended) logic. The dual operators are, respectively, $\langle\psi\rangle\phi$, $\somepub\phi$, $\somegrp{G}\phi$ and $\somecoal{G}\phi$. When talking about a specific group of agents (say $i$ and $j$) we will typically write $\allgrp{i,j}$ and $\allcoal{i,j}$ rather than $\allgrp{\{i,j\}}$ and $\allcoal{\{i,j\}}$, respectively.

We let PAL refer to EL augmented with public announcements, 
 APAL  refer to EL augmented with arbitrary public announcements,
GAL  refer to EL augmented with arbitrary group announcements, and
CAL refer to EL augmented with coalition announcements.

In order to elegantly present the semantics of these logics, later, we further define the sublanguage ${EL}^G$ to be the set of formulas of the type 
$\bigwedge_{a\in G}\know_a\phi_a$, where for each $a\in G$, $\phi_a\in{EL}$.

Formulas of these logics are interpreted over structures \(M = (\States, \sim, V)\), where 
\(\States\) is a non-empty set of worlds, 
\(\sim:\Agents\longrightarrow \wp(\States\times \States)\) assigns a reflexive, transitive and symmetric accessibility relation \(\sim_a\) to each agent \(a\) (in other words, an equivalence relation), and 
\(V: P\longrightarrow\wp(\States)\) maps each proposition to the set of worlds where it is true. 
For each $a\in\Agents$ and each $s\in\States$, the set of worlds $\{t\ |\ s\sim_a t\}$, for which we write $[s]_a$, is an equivalence class. A pair $M_s$, where $s$ is a world in the domain of $M$, is called a {\em pointed model}.

Let \(M = (\States, \sim, V)\) and suppose that \(s\in S\). The semantics of EL and the operators above are given recursively for a pointed model $M_s$:
$$\begin{array}{ll}
M_s\models p &{\rm iff}\ s\in V(p)\\
M_s\models\lnot\phi &{\rm iff}\ M_s\not\models\phi\\
M_s\models\phi_1\land\phi_2&{\rm iff}\ M_s\models\phi_1\ {\rm and}\ M_s\models\phi_2\\
M_s\models \know_a\phi &{\rm iff}\ \forall t\in S\  {\rm where}\ s\sim_a t,\ M_t\models\phi\\
M_s\models [\psi]\phi &{\rm iff}\ M_s\models\psi\ {\rm implies}\ M^{\psi}_s\models\phi\\
M_s\models \allpub\phi &{\rm iff}\ \forall \psi\in {EL},\ M_s\models[\psi]\phi\\
M_s\models \allgrp{G}\phi &{\rm iff}\ \forall \psi\in{EL}^G,\ M_s\models[\psi]\phi\\
M_s\models \allcoal{G}\phi &{\rm iff}\ \forall \psi\in{EL}^G,\ \exists \psi'\in{EL}^{\Agents\backslash G},\textrm{ s.t. } M_s\models[\psi\land\psi']\phi
\end{array}
$$
where \(M^{\psi} = (S',\sim', V')\) is such that: 
\(S' = \{s\in S\ |\ M_s\models\psi\}\); 
for all \(a\in A\), \(\sim'_a \ = \ \sim_a\cap (S'\times S')\); 
and for all \(p\in P\), \(V'(p) = V(p)\cap S'\).  

We say that a formula \(\phi\) is satisfiable if there exists some model \(M= (S, \sim, V)\) 
and some world \(s\in S\) such that \(M_s\models\phi\), and if \(M_s\models\phi\) for all 
model-world pairs $M_s$ we say \(\phi\) is valid.

The formula $[\psi]\phi$ expresses the property that after the true announcement of $\psi$, $\phi$ will hold. 
If $\psi$ is not true (in a given world) then the world is not consistent with the announcement of $\psi$, so $[\psi]\phi$ is deemed to be vacuously true in such a world. We note that $\langle\psi\rangle\phi$ abbreviates $\neg[\psi]\neg\phi$, which has the same interpretation except that when $\psi$ is not true in a given world, its interpretation defaults to false. 
It is known that epistemic logic extended with public announcements is expressively equivalent to the base epistemic logic, and as such it is decidable \cite{hvdetal.del:2007}.

The formula \(\allpub\phi\) expresses the statement {\em ``after publicly announcing any true formula of epistemic logic, $\phi$ must be true.''} Now, firstly, we can replace `any true formula' by `any formula': the restriction that the formula is true in the state of evaluation can be dropped, because if the formula is false the announcement cannot be made, so that $\allpub\phi$ is then vacuously true. Secondly, this statement implicitly quantifies over all formulas $\psi$ of epistemic logic, which is an infinite set and, expressively, a very powerful device. We give an example of this expressive power below.
\begin{example}
Suppose $\phi$ is the formula $\susp_b p\rightarrow \susp_a p$. The formula \(\allpub\phi\) is true at some world $s$ where \(p\) is not true, if and only if for every \(b\)-related world \(u\), {\em for every epistemic formula} \(\psi\), if $u$ satisfies $p\et\psi$ then there must be some \(a\)-related world, \(v_\psi\), that also satisfies $p\et\psi$. Otherwise, announcing $\psi$ would be sufficient to ensure $\susp_b p \land \know_a \neg p$. Now, $\psi$ is an arbitrary formula, so it can be arbitrarily complex. For each such $\psi$, the world $v_\psi$ may be different, but the more complex $\psi$ becomes, the more similar $u$ and $v_\psi$ must be, as they correspond on $\psi$. Thus the formula $\allpub\phi$ is enough to say that every $b$-related world that does satisfies $p$ can be ``approximated arbitrarily well'' by an $a$-related world. This scenario is informally illustrated in Figure~\ref{finApprox} and we will make this notion of approximation more precise in the Section~\ref{sec.grid}.
In \cite{frenchetal:2008} it was shown that such a property could be exploited to encode a recursively enumerable tiling problem.
\label{finApproxEG}
\end{example}

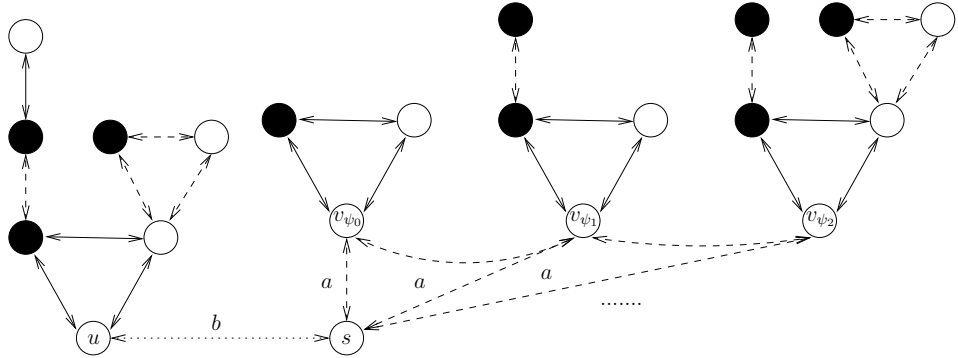
\begin{figure}
\begin{center}
\scalebox{0.7}{
\begin{picture}(0,0)%
\includegraphics{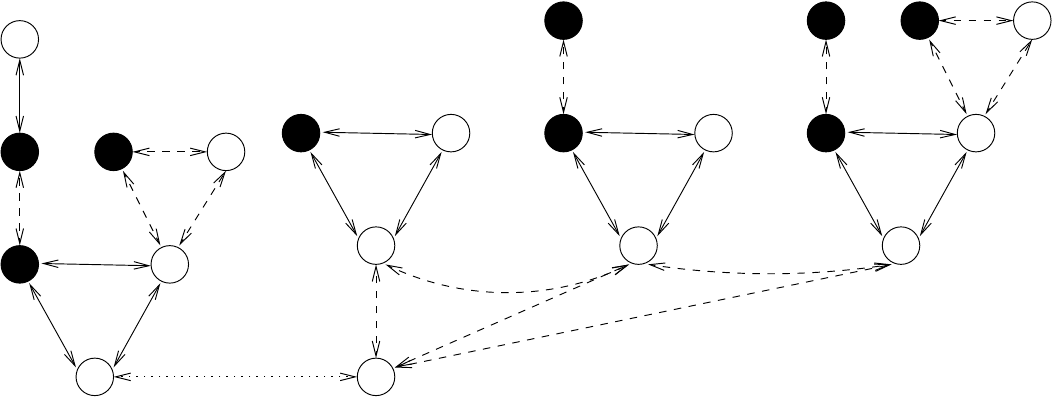}%
\end{picture}%
\setlength{\unitlength}{3947sp}%
\begingroup\makeatletter\ifx\SetFigFont\undefined%
\gdef\SetFigFont#1#2#3#4#5{%
  \fontsize{#1}{#2pt}%
  \fontfamily{#3}\fontseries{#4}\fontshape{#5}%
  \selectfont}%
\fi\endgroup%
\begin{picture}(8416,3226)(2243,-8980)
\put(7501,-8461){\makebox(0,0)[lb]{\smash{{\SetFigFont{12}{14.4}{\rmdefault}{\mddefault}{\updefault}{\color[rgb]{0,0,0}$.......$}%
}}}}
\put(5026,-8311){\makebox(0,0)[lb]{\smash{{\SetFigFont{12}{14.4}{\rmdefault}{\mddefault}{\updefault}{\color[rgb]{0,0,0}$a$}%
}}}}
\put(5851,-8311){\makebox(0,0)[lb]{\smash{{\SetFigFont{12}{14.4}{\rmdefault}{\mddefault}{\updefault}{\color[rgb]{0,0,0}$a$}%
}}}}
\put(6976,-8236){\makebox(0,0)[lb]{\smash{{\SetFigFont{12}{14.4}{\rmdefault}{\mddefault}{\updefault}{\color[rgb]{0,0,0}$a$}%
}}}}
\put(4051,-8686){\makebox(0,0)[lb]{\smash{{\SetFigFont{12}{14.4}{\rmdefault}{\mddefault}{\updefault}{\color[rgb]{0,0,0}$b$}%
}}}}
\put(2951,-8811){\makebox(0,0)[lb]{\smash{{\SetFigFont{12}{14.4}{\rmdefault}{\mddefault}{\updefault}{\color[rgb]{0,0,0}$u$}%
}}}}
\put(5201,-8811){\makebox(0,0)[lb]{\smash{{\SetFigFont{12}{14.4}{\rmdefault}{\mddefault}{\updefault}{\color[rgb]{0,0,0}$s$}%
}}}}
\put(5131,-7711){\makebox(0,0)[lb]{\smash{{\SetFigFont{12}{14.4}{\rmdefault}{\mddefault}{\updefault}{\color[rgb]{0,0,0}$v_{\psi_0}$}%
}}}}
\put(7231,-7711){\makebox(0,0)[lb]{\smash{{\SetFigFont{12}{14.4}{\rmdefault}{\mddefault}{\updefault}{\color[rgb]{0,0,0}$v_{\psi_1}$}%
}}}}
\put(9331,-7711){\makebox(0,0)[lb]{\smash{{\SetFigFont{12}{14.4}{\rmdefault}{\mddefault}{\updefault}{\color[rgb]{0,0,0}$v_{\psi_2}$}%
}}}}
\end{picture}%
}
\end{center}
\caption{An informal representation of the scenario in Example~\ref{finApproxEG}. The $\phi$-world is the root of the structure, the link from the root to the $b$-accessible world is dotted, whereas the links to $a$-accessible worlds are dashed. As the $\psi_i$ become increasingly complex, the structures attached to $v_{\psi_i}$ (from left to right) become an increasingly accurate approximation of the real structure in $u$.}\label{finApprox}
\end{figure}

The formula $\allgrp{G}\phi$ expresses the property that after a group of agents (simultaneously) announce any statements that they know to be true, $\phi$ will be true. This statement has an implicit quantifier in it as well, but this time it is only quantifying over formulas that agents know to be true (in base epistemic logic), rather than over all formulas that are true. The example above does not apply, as the worlds $v$ and $u$ would now need to agree only on formulas that some agent in $G$ knows to be true, rather than all formulas. As such we can not claim that $u$ can be ``approximated arbitrarily well'', since the set of formulas known by some agent in $G$ creates an upper bound for this ``approximation''. 
However, we are able to exploit a different type of property with group announcements. 
\begin{example}\label{grFinApprox}
A group announcement made by a single agent must preserve all worlds that are indistinguishable by that agent to the current world. Therefore, supposing that an atom $p$ is not true at the current world, $s$, and consider formula $\allgrp{b}\susp_a p$ to be true there. This means that after every possible announcement that agent $b$ knows to be true, agent $a$ still considers a world where $p$ is true to be possible. Therefore, there is no simple epistemic formula that $b$ knows to be true, and that is false at every $a$-reachable world where $p$ is true. Suppose that we are able to characterize the set of all worlds that $b$ is unable to distinguish from $s$ up to a given level of complexity, using an epistemic formula, $\psi$. It follows that  $\psi$ will be known to agent $b$, and that after the announcement of $\psi$, agent $a$ will still consider a world where $p$ is true to be possible. Therefore, it must be that there is an $a$-reachable world that agrees with one of the $b$-reachable worlds up to the given level of complexity. As the level of complexity of $\psi$ is arbitrary, again we are able to assume that every $b$-reachable world where $p$ is true may be ``approximated arbitrarily well'' by some $a$-reachable world. We will formalise this argument in Section~\ref{sec.grid}.
\end{example}

Finally in the case of coalition announcements there is an alternation in the announcement made by, and known by, the coalition, and an announcement made by, and known by, those not in the coalition. This alternation makes the constructions we require much more difficult. However, we are able to avoid alternations by only using ``coalitions'' that are the entire set of agents. Consequently, the coalition announcement is semantically equivalent to  a group announcement where the group is the full set of agents. We have chosen to present the proof for the undecidability of CAL separately from the proof of undecidability of GAL, since the proof for GAL is more elegant and only relies on single agent groups. However, the underlying mechanism is similar for both cases.

\section{Tilings}\label{sec.tilings} 

The undecidability of a logic may be shown by encoding an undecidable tiling problem into the logic.
The tiling problem we will use is as follows:
\begin{definition}\label{def:tiling}
Let \(C\) be a finite set of {\em colours} and define a \(C\)-tile to be a four-tuple over \(C\), 
\(\gamma = (\gamma^{u},\gamma^{r},\gamma^d,\gamma^{\ell})\), where the elements are referred to as, respectively, 
{\em up, right, down} and {\em left}. 
The tiling problem is, 
for any given finite set of \(C\)-tiles, \(\Gamma\), determine if there is a function 
\(\lambda:\N\times\N\longrightarrow\Gamma\) such that  for all \((i, j)\in\N\times\N\):
\begin{enumerate}
\item \(\lambda(i,j)^u = \lambda(i,j+1)^d\)
\item \(\lambda(i,j)^r = \lambda(i+1,j)^{\ell}\).
\end{enumerate}
\end{definition}

The tiling problem is known to be co-RE complete \cite{berger:1966,harel:1986}. 
In \cite{frenchetal:2008,agotnesetal:2014} it was shown, given a set of tiles $\Gamma$, that we could define a formula of APAL that was satisfiable if and only if $\Gamma$ could tile $\N\times\N$.
In this paper we take a similar, but improved, approach. 
The proof presented in \cite{frenchetal:2008} was complicated because the grid-like structure could only be enforced up to $n$ bisimilarity (for an arbitrary $n$). For GAL and CAL the restriction to arbitrary group announcements makes it harder to capture the notion of $n$-bisimilarity, which via the correspondence theory of $\cite{vanBenthem:1984}$ is equivalent to two worlds agreeing on the interpretation of all formulas up to a given modal depth. As  GAL and CAL  do not quantify over all formulas, but rather only over the formulas that are known by the group of agents, this approach is more challenging.

\section{$n$-Bisimulation}\label{sec.bisim}
\newcommand{\nbisim}[2]{\ensuremath{\|#1\|_{#2}}}
A key concept we require in this proof is $n$-bisimilarity:
\begin{definition}\label{nbisim}
Fix a finite set of atoms, $\Pi$.
Given a model, $M = (S,\sim, V)$, an $n$-$\Pi$-bisimulation over $M$ is a relation $R_n\subseteq S\times S$ defined recursively as:
\begin{enumerate}
\item $s R_0 t$ if and only if for all $p\in\Pi$, $s\in V(p)$ if and only if $t\in V(p)$.
\item for all $m>0$, $s R_{m} t$ if and only if $s R_{m-1} t$ and:
	\begin{enumerate}
	\item for all $i\in\Agents$, for all $u$ where $s\sim_i u$, there is some $v$ where $t\sim_i v$ and $u R_m v$;
	\item for all $i\in\Agents$, for all $v$ where $t\sim_i v$, there is some $u$ where $s\sim_i u$ and $v R_m u$.
	\end{enumerate}
\end{enumerate}
Given such a model $M$, and some $s\in S$, we let $\nbisim{s}{n}^\Pi\subseteq S$ be the set of worlds $n$-$\Pi$-bisimilar to $s$. (We omit the $\Pi$ when it is clear from context). 
\end{definition}
It is clear that $\nbisim{s}{n}$ is an equivalence class for all $s\in S$.
Another important property of $n$-bisimilarity is that any two worlds, $s$ and $t$, that are not $n$-bisimilar must have a witnessing formula $\phi\in\lang_{el}$, such that $M_s\models\phi$ and $M_t\not\models\phi$.
Thus, if we are able to establish that two worlds in the model {\em cannot} be distinguished by any formula, it follows that these two worlds are $n$-bisimilar for all $n$. Through this property we are able to constrain a model to be bisimilar to a grid-like structure, described in the following section.

\begin{lemma}\label{nbisimwitness}
Let $\Pi$ be a finite set of propositional atoms. Suppose that $M = (S,\sim, V)$ and $s\in S$. Then for all  $n$, there is some $\lang_{el}$ formula, $\phi$, such that for all $t\in S$, $M_t\models\phi$ if and only if $t\in\nbisim{s}{n}^\Pi$.
\end{lemma}
\begin{proof}
We show this by induction over $n$, where the induction hypothesis is, for any given $n$:
\quote{$IH(n)$: for all $t\in S$, there is some formula, $\phi^t_n$, that is sufficient to distinguish all states that are not $n$-$\Pi$-bisimilar to $t$.}
That is $M_t\models\phi^t_n$ and if $u\notin\nbisim{t}{n}$ then $M_u\not\models\phi^t_n$.

For the base case, $IH(0)$, it is clear that if $u\notin\nbisim{t}{0}$, then $t$ and $u$ must disagree on the interpretation of some atom $p\in \Pi$. 
Therefore we may set $\phi^t_0 = \bigwedge\{p\in\Pi\ |\ t\in V(p)\}\et\bigwedge\{\neg p\in\Pi\ |\ t\notin V(p)\}$.
As $\Pi$ is finite it is clear that this is a well-defined formula satisfying the base of this induction.

Now suppose $IH(m)$ holds, so for all $t\in S$, there is a formula $\phi^t_m$ such that $M_t\models\phi^t_m$ and if $u\notin\nbisim{t}{m}$ then $M_u\not\models\phi^t_m$. Now for $u, t\in S$ such that  $u\notin\nbisim{t}{m+1}$, we have three possible scenarios:
\begin{enumerate}
\item $u\notin\nbisim{t}{m}$, in which case there is some $\lang_{el}$ formula $\phi^t_m$ where $M_t\models\phi^t_m$ and $M_u\not\models\phi^t_m$. 
\item there is some $i\in\Agents$, and some $v$ where $t\sim_i v$, but for all $v'$ where $u\sim_i v'$, $v'\notin\nbisim{v}{m}$. Therefore, for every such $v$ there is some formula $\phi^v_m$ where $M_v\models\phi^v_m$ and $M_{v'}\not\models\phi^v_m$.
As these formulas are taken from a finite set, there is a finite formula $\psi^t_m = \susp_i(\bigvee_{t\sim_i v}\phi^v_m)$ such that $M_t\models\psi^t_m$ and $M_u\not\models\psi^t_m$.
\item there is some $i\in\Agents$ and some $v'$ where $u\sim_i v'$, but for all $v$ where $t\sim_i v$, $v'\notin\nbisim{v}{m}$. 
As $v'\notin\nbisim{v}{m}$ for all $v$ where $t\sim_i v$, it must be the case that $M_{v'}\not\models\bigvee_{t\sim_i v}\phi^v_m$. 
We let $\xi^t_m = \know_i(\bigvee_{t\sim_i v}\phi^v_m)$, and it is clear that $M_t\models\xi^t_m$.   
\end{enumerate}
We define the $\lang_{el}$ formula $\phi^t_{m+1} = \phi^t_m\et\psi^t_m\et\xi^t_m$. From the reasoning above it is clear that $u\in\nbisim{t}{m+1}$ if and only if $M_u\models\phi^t_{m+1}$. 
\end{proof}

From the proof of Lemma~\ref{nbisimwitness} we can see that given a finite set of propositions, $\Pi$, we only require a finite set of formulas to distinguish all states, up to $n$-$\Pi$-bisimilarity.
\begin{corollary}\label{finitenbisim}
Let $\Pi$ be a finite set of propositional atoms and $n\in\N$. Then there is a finite set of formulas $\phi_0,\phi_1,\hdots,\phi_n$ such that for all $M = (S, R, V)$, for all $u,v\in S$:
\begin{enumerate}
\item there is some $i\leq n$ such that $M_u\models\phi_i$; and
\item if there is some $i\leq n$ such that $M_u\models\phi_i$ and $M_v\models\phi_i$, then $v\in\nbisim{u}{n}^\Pi$. 
\end{enumerate}
\end{corollary}
\begin{proof}
This follows immediately from the proof of Lemma~\ref{nbisimwitness}.
\end{proof}

Finally, we present a well known folk-lemma that $n$-$\Pi$-bisimulations preserve the interpretation of pure modal formulas over $\Pi$, with modal depth less than or equal to $n$.
\begin{lemma}\label{modDepPres}
Let $M = (S, R, V)$ and suppose that $s,t\in S$ such that $t\in\nbisim{s}{n}^\Pi$. Then for all pure modal formulas $\phi$, consisting of only atoms in $\Pi$, that have modal depth less than or equal to $n$, we have $M_s\models\phi$ if and only if $M_t\models\phi$. 
\end{lemma}

\begin{proof}
This is shown by a simple induction over $n$. 
If $n=0$, then $s$ and $t$ must agree on the interpretation of all propositional atoms in $\Pi$, and so must also agree on the interpretation of all propositional formulas.
For the induction, suppose that for all formulas $\phi$ consisting of atoms in $\Pi$, and having modal depth less than $n$, for all $u,v\in S$ where $u\in\nbisim{v}{n-1}$, we have $M_u\models\phi$ if and only if $M_v\models\phi$.
Let $s, t\in S$ be such that $t\in\nbisim{s}{n}$, and suppose that $\phi$ is any formula with modal depth at most $n$. 
The formula $\phi$ can be broken down into a propositional combination of propositional atoms, $p_0,\hdots,p_a$ in $\Pi$ and modal formulas $\phi_0,\hdots,\phi_b$.
By definition, $t\in\nbisim{s}{0}$ so for $x\leq a$, $M_t\models p_a$ if and only if $M_s\models p_a$.
Also, for every $x\leq b$ suppose that $\phi_x = \know_i\psi$, where the modal depth of $\psi$ is less than $n$. 
If $M_t\models\know_i\psi$, then for every $t'\in t R_i$, $M_{t'}\models\psi$. 
Since $t\in \nbisim{s}{n}$, for every $s'\in s R_i$, there is some $t'\in t R_i$ such that $t'\in\nbisim{s'}{n-1}$. 
By the induction hypothesis, $M_{t'}\models\psi$ if and only if $M_{s'}\models\psi$, so for every $s'\in s R_i$, we have $M_{s'}\models\psi$, and $M_s\models\know_i\psi$.
Conversely, if $M_t\not\models\know_i\psi$, then there is some $t'\in t R_i$ where $M_{t'}\not\models\psi$.
Since $t\in \nbisim{s}{n}$, for every $t'\in t R_i$, there is some $s'\in s R_i$ such that $t'\in\nbisim{s'}{n-1}$.
By the induction hypothesis, and the fact that the modal depth of $\psi$ is less than $n$, it follows that $M_{s'}\not\models\psi$. 
Therefore, $M_s\not\models\know_i\psi$, and thus $M_t\models\phi_x$ if and only if $M_s\models\phi_x$.
As $s$ and $t$ agree on the interpretations of all the parts of $\phi$, they must also agree on the interpretation of $\phi$, completing the inductive case.
\end{proof}

\newcommand{\east}{\ensuremath{\mathfrak{e}}}
\newcommand{\west}{\ensuremath{\mathfrak{w}}}
\newcommand{\north}{\ensuremath{\mathfrak{n}}}
\newcommand{\south}{\ensuremath{\mathfrak{s}}}
\newcommand{\trans}{\ensuremath{\mathfrak{t}}}
\newcommand{\hearts}{\ensuremath{{\heartsuit}}}
\newcommand{\clubs}{\ensuremath{{\clubsuit}}}
\newcommand{\diamonds}{\ensuremath{{\diamondsuit}}}
\newcommand{\spades}{\ensuremath{{\spadesuit}}}
\newcommand{\edge}{\ensuremath{\mathfrak{e}}}
\newcommand{\sq}{\ensuremath{\mathfrak{s}}}
\newcommand{\up}{\ensuremath{u}}
\newcommand{\down}{\ensuremath{d}}
\newcommand{\lef}{\ensuremath{\ell}}
\newcommand{\righ}{\ensuremath{r}}
\newcommand{\centr}{\ensuremath{c}}

\section{The grid-like structure}\label{sec.grid}

In this section we provide the main construction that allows the tiling problem to be reduced to the satisfiability problem for formulas in APAL, GAL and CAL. The main steps of this process are:
\begin{enumerate}
\item enforcing the structure of a satisfying model to have a grid-like structure; 
\item defining a formula to represent common knowledge; 
\item using propositional atoms to represent tiles, express the formula ``it is common knowledge that adjacent tiles on the grid have matching sides.''
\end{enumerate}

The model we aim to build represents an infinite grid, or a checkerboard. Rather than the typically black and white squares of a checkerboard we use the card suits $\hearts$, $\clubs$, $\diamonds$ and $\spades$ to label the squares, in a regular fashion: rows either alternate between $\spades$ and $\diamonds$ squares, or $\hearts$ and $\clubs$ squares; and columns either alternate between $\spades$ and $\hearts$ squares, or $\diamonds$ and $\clubs$ squares. We will provide a formula that enforces these constraints on the model, and also requires that the squares are labelled with tiles as specified by the tiling problem. The satisfiability of this formula is then equivalent to knowing whether a set of tiles can tile the plane. 
The base structure is represented in Figure~\ref{checkerboard}.
\begin{figure}
\begin{center}
\scalebox{0.7}{
\begin{picture}(0,0)%
\includegraphics{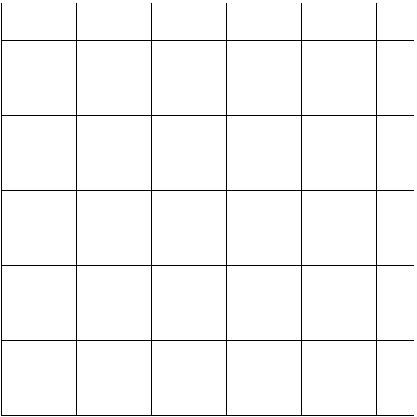}%
\end{picture}%
\setlength{\unitlength}{3947sp}%
\begingroup\makeatletter\ifx\SetFigFont\undefined%
\gdef\SetFigFont#1#2#3#4#5{%
  \reset@font\fontsize{#1}{#2pt}%
  \fontfamily{#3}\fontseries{#4}\fontshape{#5}%
  \selectfont}%
\fi\endgroup%
\begin{picture}(3324,3324)(1189,-7573)
\put(3826,-7336){\makebox(0,0)[lb]{\smash{{\SetFigFont{12}{14.4}{\familydefault}{\mddefault}{\updefault}{\color[rgb]{0,0,0}$\hearts$}%
}}}}
\put(1426,-7336){\makebox(0,0)[lb]{\smash{{\SetFigFont{12}{14.4}{\familydefault}{\mddefault}{\updefault}{\color[rgb]{0,0,0}$\hearts$}%
}}}}
\put(1426,-6736){\makebox(0,0)[lb]{\smash{{\SetFigFont{12}{14.4}{\rmdefault}{\mddefault}{\updefault}{\color[rgb]{0,0,0}$\spades$}%
}}}}
\put(2026,-7336){\makebox(0,0)[lb]{\smash{{\SetFigFont{12}{14.4}{\rmdefault}{\mddefault}{\updefault}{\color[rgb]{0,0,0}$\clubs$}%
}}}}
\put(1426,-6136){\makebox(0,0)[lb]{\smash{{\SetFigFont{12}{14.4}{\familydefault}{\mddefault}{\updefault}{\color[rgb]{0,0,0}$\hearts$}%
}}}}
\put(3226,-6136){\makebox(0,0)[lb]{\smash{{\SetFigFont{12}{14.4}{\rmdefault}{\mddefault}{\updefault}{\color[rgb]{0,0,0}$\clubs$}%
}}}}
\put(3226,-7336){\makebox(0,0)[lb]{\smash{{\SetFigFont{12}{14.4}{\rmdefault}{\mddefault}{\updefault}{\color[rgb]{0,0,0}$\clubs$}%
}}}}
\put(1426,-5536){\makebox(0,0)[lb]{\smash{{\SetFigFont{12}{14.4}{\rmdefault}{\mddefault}{\updefault}{\color[rgb]{0,0,0}$\spades$}%
}}}}
\put(1426,-4936){\makebox(0,0)[lb]{\smash{{\SetFigFont{12}{14.4}{\familydefault}{\mddefault}{\updefault}{\color[rgb]{0,0,0}$\hearts$}%
}}}}
\put(2026,-4936){\makebox(0,0)[lb]{\smash{{\SetFigFont{12}{14.4}{\rmdefault}{\mddefault}{\updefault}{\color[rgb]{0,0,0}$\clubs$}%
}}}}
\put(2626,-4936){\makebox(0,0)[lb]{\smash{{\SetFigFont{12}{14.4}{\familydefault}{\mddefault}{\updefault}{\color[rgb]{0,0,0}$\hearts$}%
}}}}
\put(3226,-4936){\makebox(0,0)[lb]{\smash{{\SetFigFont{12}{14.4}{\rmdefault}{\mddefault}{\updefault}{\color[rgb]{0,0,0}$\clubs$}%
}}}}
\put(3226,-5536){\makebox(0,0)[lb]{\smash{{\SetFigFont{12}{14.4}{\rmdefault}{\mddefault}{\updefault}{\color[rgb]{0,0,0}$\diamonds$}%
}}}}
\put(2026,-6136){\makebox(0,0)[lb]{\smash{{\SetFigFont{12}{14.4}{\rmdefault}{\mddefault}{\updefault}{\color[rgb]{0,0,0}$\clubs$}%
}}}}
\put(2026,-6736){\makebox(0,0)[lb]{\smash{{\SetFigFont{12}{14.4}{\rmdefault}{\mddefault}{\updefault}{\color[rgb]{0,0,0}$\diamonds$}%
}}}}
\put(2026,-5536){\makebox(0,0)[lb]{\smash{{\SetFigFont{12}{14.4}{\rmdefault}{\mddefault}{\updefault}{\color[rgb]{0,0,0}$\diamonds$}%
}}}}
\put(2626,-5536){\makebox(0,0)[lb]{\smash{{\SetFigFont{12}{14.4}{\rmdefault}{\mddefault}{\updefault}{\color[rgb]{0,0,0}$\spades$}%
}}}}
\put(2626,-6136){\makebox(0,0)[lb]{\smash{{\SetFigFont{12}{14.4}{\familydefault}{\mddefault}{\updefault}{\color[rgb]{0,0,0}$\hearts$}%
}}}}
\put(2626,-6736){\makebox(0,0)[lb]{\smash{{\SetFigFont{12}{14.4}{\rmdefault}{\mddefault}{\updefault}{\color[rgb]{0,0,0}$\spades$}%
}}}}
\put(2626,-7336){\makebox(0,0)[lb]{\smash{{\SetFigFont{12}{14.4}{\familydefault}{\mddefault}{\updefault}{\color[rgb]{0,0,0}$\hearts$}%
}}}}
\put(3226,-6736){\makebox(0,0)[lb]{\smash{{\SetFigFont{12}{14.4}{\rmdefault}{\mddefault}{\updefault}{\color[rgb]{0,0,0}$\diamonds$}%
}}}}
\put(3826,-6736){\makebox(0,0)[lb]{\smash{{\SetFigFont{12}{14.4}{\rmdefault}{\mddefault}{\updefault}{\color[rgb]{0,0,0}$\spades$}%
}}}}
\put(3826,-6136){\makebox(0,0)[lb]{\smash{{\SetFigFont{12}{14.4}{\familydefault}{\mddefault}{\updefault}{\color[rgb]{0,0,0}$\hearts$}%
}}}}
\put(3826,-5536){\makebox(0,0)[lb]{\smash{{\SetFigFont{12}{14.4}{\rmdefault}{\mddefault}{\updefault}{\color[rgb]{0,0,0}$\spades$}%
}}}}
\put(3826,-4936){\makebox(0,0)[lb]{\smash{{\SetFigFont{12}{14.4}{\familydefault}{\mddefault}{\updefault}{\color[rgb]{0,0,0}$\hearts$}%
}}}}
\end{picture}%
}
\end{center}
\caption{The checkerboard configuration.}\label{checkerboard}
\end{figure}

To represent this structure in an epistemic structure, $M = (S, \sim, V)$, we will suppose that each square in the grid corresponds to five worlds in $S$, labelled: {\em up, down, left, right and center} (corresponding respectively to the atoms $\up,\ \down,\ \lef,\ \righ$ and the center square will be labelled by exactly one of the propositions $\hearts,\ \clubs,\ \diamonds$ or $\spades$). We also suppose that there are just two agents: the  {\em square} agent, $\sq$; and the {\em edge} agent, $\edge$. For any square, agent $\sq$ cannot distinguish between any of the five worlds that correspond to that square. For two adjacent squares (say $s$ which is to the left of $t$) the agent $\edge$ is unable to distinguish between the right world corresponding to $s$ and the left world corresponding to $t$. The edge agent is likewise unable to distinguish between the up and down worlds corresponding the vertically adjacent tiles. Finally we suppose that the edge agent is unable to distinguish between any of the worlds that are labelled as center.
This representation is given in Figure~\ref{cbRep}.

\begin{figure}
\begin{center}
\scalebox{0.64}{
\begin{picture}(0,0)%
\includegraphics{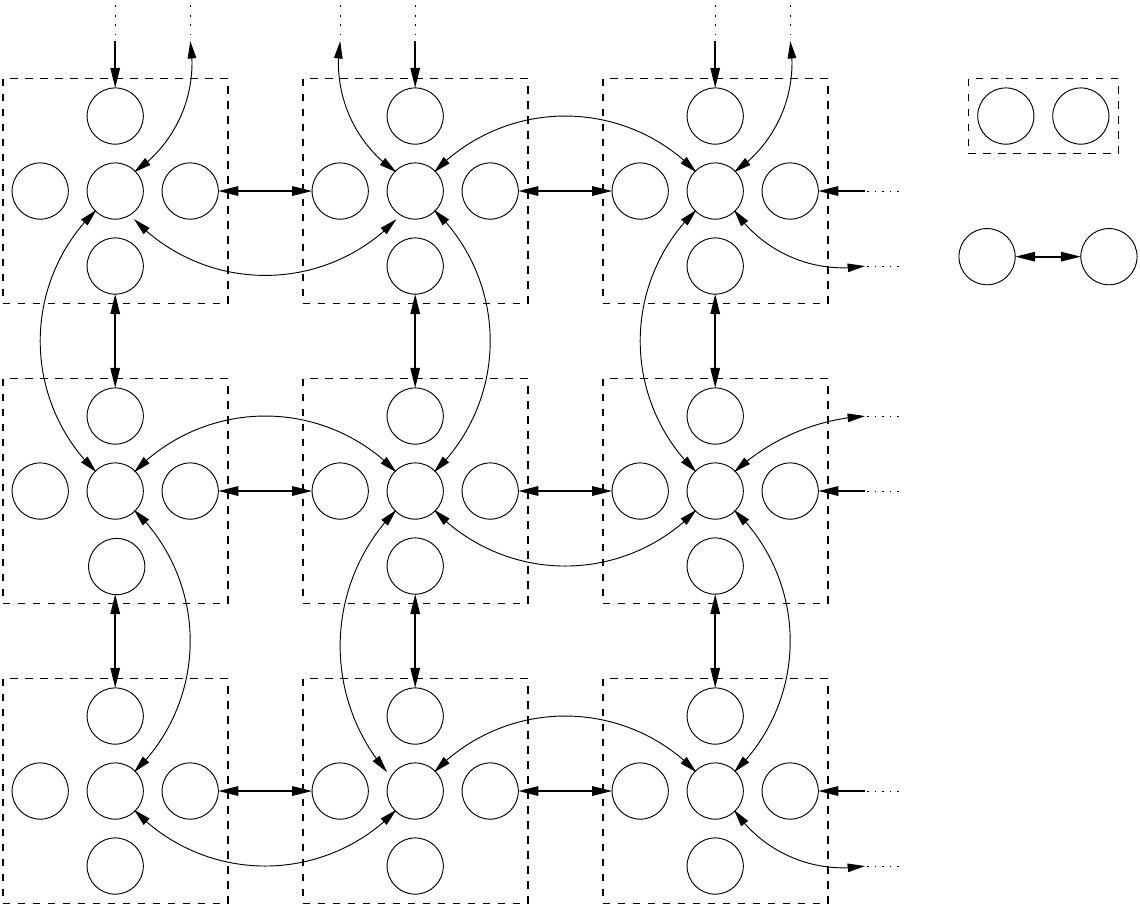}%
\end{picture}%
\setlength{\unitlength}{3947sp}%
\begingroup\makeatletter\ifx\SetFigFont\undefined%
\gdef\SetFigFont#1#2#3#4#5{%
  \reset@font\fontsize{#1}{#2pt}%
  \fontfamily{#3}\fontseries{#4}\fontshape{#5}%
  \selectfont}%
\fi\endgroup%
\begin{picture}(9105,7234)(54,-8183)
\put(8146,-3436){\makebox(0,0)[lb]{\smash{{\SetFigFont{12}{14.4}{\rmdefault}{\mddefault}{\updefault}{\color[rgb]{0,0,0}$s \sim_\edge t$}%
}}}}
\put(7991,-1936){\makebox(0,0)[lb]{\smash{{\SetFigFont{12}{14.4}{\rmdefault}{\mddefault}{\updefault}{\color[rgb]{0,0,0}$s$}%
}}}}
\put(8591,-1936){\makebox(0,0)[lb]{\smash{{\SetFigFont{12}{14.4}{\rmdefault}{\mddefault}{\updefault}{\color[rgb]{0,0,0}$t$}%
}}}}
\put(8816,-3061){\makebox(0,0)[lb]{\smash{{\SetFigFont{12}{14.4}{\rmdefault}{\mddefault}{\updefault}{\color[rgb]{0,0,0}$t$}%
}}}}
\put(7841,-3061){\makebox(0,0)[lb]{\smash{{\SetFigFont{12}{14.4}{\rmdefault}{\mddefault}{\updefault}{\color[rgb]{0,0,0}$s$}%
}}}}
\put(901,-6736){\makebox(0,0)[lb]{\smash{{\SetFigFont{12}{14.4}{\rmdefault}{\mddefault}{\updefault}{\color[rgb]{0,0,0}$\up$}%
}}}}
\put(901,-4336){\makebox(0,0)[lb]{\smash{{\SetFigFont{12}{14.4}{\rmdefault}{\mddefault}{\updefault}{\color[rgb]{0,0,0}$\up$}%
}}}}
\put(826,-3136){\makebox(0,0)[lb]{\smash{{\SetFigFont{12}{14.4}{\rmdefault}{\mddefault}{\updefault}{\color[rgb]{0,0,0}$\down$}%
}}}}
\put(1426,-2536){\makebox(0,0)[lb]{\smash{{\SetFigFont{12}{14.4}{\rmdefault}{\mddefault}{\updefault}{\color[rgb]{0,0,0}$\righ$}%
}}}}
\put(901,-1936){\makebox(0,0)[lb]{\smash{{\SetFigFont{12}{14.4}{\rmdefault}{\mddefault}{\updefault}{\color[rgb]{0,0,0}$\up$}%
}}}}
\put(3301,-1936){\makebox(0,0)[lb]{\smash{{\SetFigFont{12}{14.4}{\rmdefault}{\mddefault}{\updefault}{\color[rgb]{0,0,0}$\up$}%
}}}}
\put(5701,-4336){\makebox(0,0)[lb]{\smash{{\SetFigFont{12}{14.4}{\rmdefault}{\mddefault}{\updefault}{\color[rgb]{0,0,0}$\up$}%
}}}}
\put(3301,-4336){\makebox(0,0)[lb]{\smash{{\SetFigFont{12}{14.4}{\rmdefault}{\mddefault}{\updefault}{\color[rgb]{0,0,0}$\up$}%
}}}}
\put(5701,-6736){\makebox(0,0)[lb]{\smash{{\SetFigFont{12}{14.4}{\rmdefault}{\mddefault}{\updefault}{\color[rgb]{0,0,0}$\up$}%
}}}}
\put(901,-7336){\makebox(0,0)[lb]{\smash{{\SetFigFont{12}{14.4}{\rmdefault}{\mddefault}{\updefault}{\color[rgb]{0,0,0}$\hearts$}%
}}}}
\put(3301,-7336){\makebox(0,0)[lb]{\smash{{\SetFigFont{12}{14.4}{\rmdefault}{\mddefault}{\updefault}{\color[rgb]{0,0,0}$\clubs$}%
}}}}
\put(3301,-4936){\makebox(0,0)[lb]{\smash{{\SetFigFont{12}{14.4}{\rmdefault}{\mddefault}{\updefault}{\color[rgb]{0,0,0}$\diamonds$}%
}}}}
\put(901,-4936){\makebox(0,0)[lb]{\smash{{\SetFigFont{12}{14.4}{\rmdefault}{\mddefault}{\updefault}{\color[rgb]{0,0,0}$\spades$}%
}}}}
\put(5701,-7336){\makebox(0,0)[lb]{\smash{{\SetFigFont{12}{14.4}{\rmdefault}{\mddefault}{\updefault}{\color[rgb]{0,0,0}$\hearts$}%
}}}}
\put(5701,-4936){\makebox(0,0)[lb]{\smash{{\SetFigFont{12}{14.4}{\rmdefault}{\mddefault}{\updefault}{\color[rgb]{0,0,0}$\clubs$}%
}}}}
\put(901,-2536){\makebox(0,0)[lb]{\smash{{\SetFigFont{12}{14.4}{\rmdefault}{\mddefault}{\updefault}{\color[rgb]{0,0,0}$\hearts$}%
}}}}
\put(3301,-2536){\makebox(0,0)[lb]{\smash{{\SetFigFont{12}{14.4}{\rmdefault}{\mddefault}{\updefault}{\color[rgb]{0,0,0}$\clubs$}%
}}}}
\put(5701,-2536){\makebox(0,0)[lb]{\smash{{\SetFigFont{12}{14.4}{\rmdefault}{\mddefault}{\updefault}{\color[rgb]{0,0,0}$\hearts$}%
}}}}
\put(301,-2536){\makebox(0,0)[lb]{\smash{{\SetFigFont{12}{14.4}{\rmdefault}{\mddefault}{\updefault}{\color[rgb]{0,0,0}$\lef$}%
}}}}
\put(301,-4936){\makebox(0,0)[lb]{\smash{{\SetFigFont{12}{14.4}{\rmdefault}{\mddefault}{\updefault}{\color[rgb]{0,0,0}$\lef$}%
}}}}
\put(901,-5536){\makebox(0,0)[lb]{\smash{{\SetFigFont{12}{14.4}{\rmdefault}{\mddefault}{\updefault}{\color[rgb]{0,0,0}$\down$}%
}}}}
\put(1501,-4936){\makebox(0,0)[lb]{\smash{{\SetFigFont{12}{14.4}{\rmdefault}{\mddefault}{\updefault}{\color[rgb]{0,0,0}$\righ$}%
}}}}
\put(301,-7336){\makebox(0,0)[lb]{\smash{{\SetFigFont{12}{14.4}{\rmdefault}{\mddefault}{\updefault}{\color[rgb]{0,0,0}$\lef$}%
}}}}
\put(901,-7936){\makebox(0,0)[lb]{\smash{{\SetFigFont{12}{14.4}{\rmdefault}{\mddefault}{\updefault}{\color[rgb]{0,0,0}$\down$}%
}}}}
\put(1426,-7336){\makebox(0,0)[lb]{\smash{{\SetFigFont{12}{14.4}{\rmdefault}{\mddefault}{\updefault}{\color[rgb]{0,0,0}$\righ$}%
}}}}
\put(3301,-7936){\makebox(0,0)[lb]{\smash{{\SetFigFont{12}{14.4}{\rmdefault}{\mddefault}{\updefault}{\color[rgb]{0,0,0}$\down$}%
}}}}
\put(2701,-7336){\makebox(0,0)[lb]{\smash{{\SetFigFont{12}{14.4}{\rmdefault}{\mddefault}{\updefault}{\color[rgb]{0,0,0}$\lef$}%
}}}}
\put(3301,-6736){\makebox(0,0)[lb]{\smash{{\SetFigFont{12}{14.4}{\rmdefault}{\mddefault}{\updefault}{\color[rgb]{0,0,0}$\up$}%
}}}}
\put(3901,-7336){\makebox(0,0)[lb]{\smash{{\SetFigFont{12}{14.4}{\rmdefault}{\mddefault}{\updefault}{\color[rgb]{0,0,0}$\righ$}%
}}}}
\put(3301,-5536){\makebox(0,0)[lb]{\smash{{\SetFigFont{12}{14.4}{\rmdefault}{\mddefault}{\updefault}{\color[rgb]{0,0,0}$\down$}%
}}}}
\put(2701,-4936){\makebox(0,0)[lb]{\smash{{\SetFigFont{12}{14.4}{\rmdefault}{\mddefault}{\updefault}{\color[rgb]{0,0,0}$\lef$}%
}}}}
\put(3901,-4936){\makebox(0,0)[lb]{\smash{{\SetFigFont{12}{14.4}{\rmdefault}{\mddefault}{\updefault}{\color[rgb]{0,0,0}$\righ$}%
}}}}
\put(2701,-2536){\makebox(0,0)[lb]{\smash{{\SetFigFont{12}{14.4}{\rmdefault}{\mddefault}{\updefault}{\color[rgb]{0,0,0}$\lef$}%
}}}}
\put(3301,-3136){\makebox(0,0)[lb]{\smash{{\SetFigFont{12}{14.4}{\rmdefault}{\mddefault}{\updefault}{\color[rgb]{0,0,0}$\down$}%
}}}}
\put(3826,-2536){\makebox(0,0)[lb]{\smash{{\SetFigFont{12}{14.4}{\rmdefault}{\mddefault}{\updefault}{\color[rgb]{0,0,0}$\righ$}%
}}}}
\put(5101,-2536){\makebox(0,0)[lb]{\smash{{\SetFigFont{12}{14.4}{\rmdefault}{\mddefault}{\updefault}{\color[rgb]{0,0,0}$\lef$}%
}}}}
\put(5701,-1936){\makebox(0,0)[lb]{\smash{{\SetFigFont{12}{14.4}{\rmdefault}{\mddefault}{\updefault}{\color[rgb]{0,0,0}$\up$}%
}}}}
\put(5701,-3136){\makebox(0,0)[lb]{\smash{{\SetFigFont{12}{14.4}{\rmdefault}{\mddefault}{\updefault}{\color[rgb]{0,0,0}$\down$}%
}}}}
\put(6301,-2536){\makebox(0,0)[lb]{\smash{{\SetFigFont{12}{14.4}{\rmdefault}{\mddefault}{\updefault}{\color[rgb]{0,0,0}$\righ$}%
}}}}
\put(5101,-4936){\makebox(0,0)[lb]{\smash{{\SetFigFont{12}{14.4}{\rmdefault}{\mddefault}{\updefault}{\color[rgb]{0,0,0}$\lef$}%
}}}}
\put(5701,-5536){\makebox(0,0)[lb]{\smash{{\SetFigFont{12}{14.4}{\rmdefault}{\mddefault}{\updefault}{\color[rgb]{0,0,0}$\down$}%
}}}}
\put(6301,-4936){\makebox(0,0)[lb]{\smash{{\SetFigFont{12}{14.4}{\rmdefault}{\mddefault}{\updefault}{\color[rgb]{0,0,0}$\righ$}%
}}}}
\put(5101,-7336){\makebox(0,0)[lb]{\smash{{\SetFigFont{12}{14.4}{\rmdefault}{\mddefault}{\updefault}{\color[rgb]{0,0,0}$\lef$}%
}}}}
\put(5701,-7936){\makebox(0,0)[lb]{\smash{{\SetFigFont{12}{14.4}{\rmdefault}{\mddefault}{\updefault}{\color[rgb]{0,0,0}$\down$}%
}}}}
\put(6301,-7336){\makebox(0,0)[lb]{\smash{{\SetFigFont{12}{14.4}{\rmdefault}{\mddefault}{\updefault}{\color[rgb]{0,0,0}$\righ$}%
}}}}
\put(8156,-2386){\makebox(0,0)[lb]{\smash{{\SetFigFont{12}{14.4}{\rmdefault}{\mddefault}{\updefault}{\color[rgb]{0,0,0}$s \sim_\sq t$}%
}}}}
\end{picture}%
}
\caption{A representation of the checkerboard of Figure~\ref{checkerboard} as an epistemic model. The agent $\sq$ is unable to distinguish between any worlds inside a single square, whilst the agent $\edge$ is unable to distinguish between any two worlds connected by the reflexive, transitive closure of the  arrow relation.}\label{cbRep}
\end{center}
\end{figure}
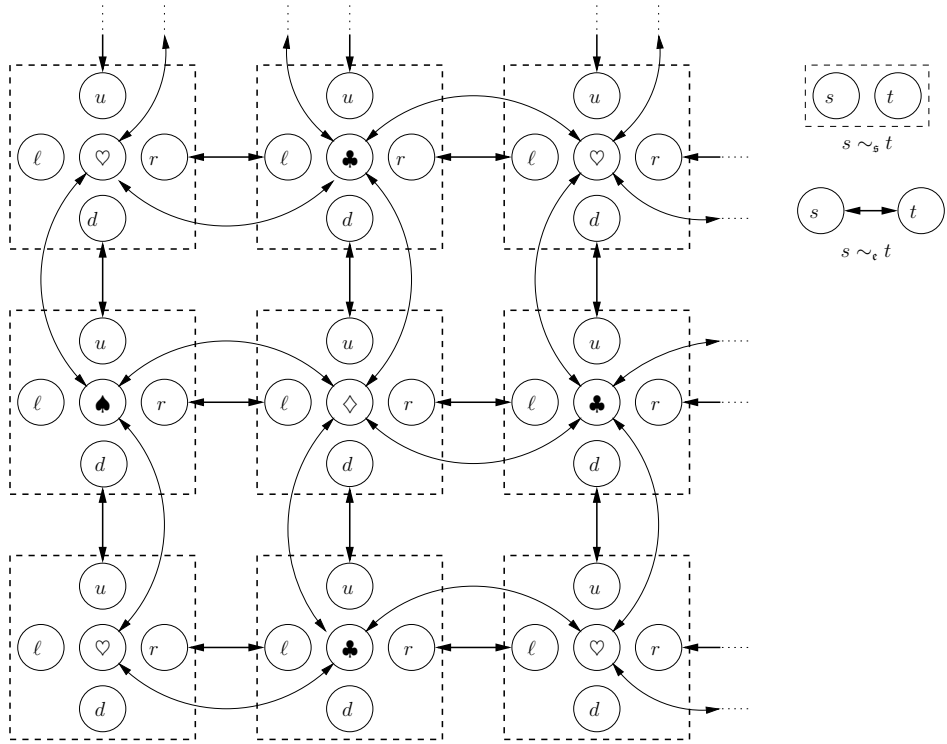

The model is infinite, so we may suppose that each square corresponds to a point in $\N\times\N$. To create a tiling, we suppose that we also have a set of propositions $C$ to label the coloured sides of tiles, and that this set is disjoint from all the other propositions introduced so far.

Defining a formula that requires a model to have such a structure (steps 1 and 2) is dependent on the separate quantifiers of APAL, GAL and CAL. We will first address the local properties of the grid and the tiling (step 3).
Suppose we have a model with a grid-like structure. Below we provide a formula such that the satisfiability of that formula at the center of a square is equivalent to the existence of a $\Gamma$ tiling of the plane. Let $\Gamma = \{\gamma_1,...\gamma_n\}$ be a set of tiles. 
We require the following properties to be true:
\begin{enumerate}
\item {\em Every world is labelled by exactly one colour:}
\begin{equation}\label{oneCol}
oneCol = \know_\sq\know_\edge\know_\sq\bigvee_{c\in C} \left(c\et \bigwedge_{d\in C\backslash\{c\}}\neg d\right)
\end{equation}
Since every square contains a center world connected to every other center world by the accessibility relation for $\edge$, $\know_\sq\know_\edge\know_\sq\phi$ implies that $\phi$ must be true in all worlds. In this case we require that every world is labelled by exactly one colour, and this is specified via propositional reasoning. 
\item {\em Every square corresponds to a tile in the set }$\Gamma$.
\begin{equation}\label{tile-gamma}
tile_\Gamma = \know_\sq\know_\edge\know_\sq\left(\bigvee_{\gamma\in\Gamma}\bigwedge\begin{array}{c} \up\imp\gamma^\up\\\righ\imp\gamma^\righ\\\down\imp\gamma^\down\\\lef\imp\gamma^\lef\end{array}\right)
\end{equation}
As above $\know_\edge\phi$ requires $\phi$ to be true at every center world (at least). 
The nested $\know_\sq$ operator requires that every square corresponds to some tile in $\Gamma$, by requiring that for some tile $\gamma$, the $\up$-world has the same colour as the top of $\gamma$, the $\righ$-world has the same colour as the right side of $\gamma$, the $\down$-world has the same colour as the bottom of $\gamma$ and the $\lef$-world has the same colour as the left side of $\gamma$. 
\item {\em The agent $\edge$ always knows the colour.}
\begin{equation}\label{match}
match = \know_\sq\know_\edge\know_\sq\bigvee_{c\in C}(\know_\edge c)
\end{equation}
As the $\edge$-agent cannot distinguish the worlds corresponding to adjoining sides of adjacent squares, requiring that agent $\edge$ always knows the colour means that any worlds corresponding to adjoining sides of adjacent squares must have the same colour, as required in the tiling problem. Note that a consequence of this formula is that the central nodes will all be labelled by the same colour. However, as we are only concerned with the colour labelling at the edge of the tiles, this does not effect the encoding. 
\end{enumerate}
We combine these properties in the formula:
\begin{equation}\label{tiling}
SAT_\Gamma = oneCol\et tile_\Gamma\et match
\end{equation}

We note that $SAT_\Gamma$ does not contain any arbitrary group announcement operators, and is in fact a formula of $\lang_{el}$.
It is also clear that the formula $Tile_\Gamma$ is satisfiable by some models that do not have the grid-like structure. 
Our next task is to ensure that the worlds in each square are labelled correctly, and adjacent squares have corresponding labels.
To this end, let $\Lambda = \{\up,\down,\lef,\righ,\hearts,\clubs,\diamonds,\spades\}$ be the set of labels. We require the following properties:
\begin{enumerate}
\item {\em Every world satisfies exactly one label $\lambda\in \Lambda$.}
\begin{equation}\label{oneLabel}
oneLabel =\know_\sq\know_\edge\know_\sq\left(\bigvee_{\lambda\in \Lambda}\left(\lambda\et\bigwedge_{\mu\in\Lambda\backslash\{\lambda\}}\neg\mu\right)\right)
\end{equation}
The reasoning here is analogous to the formula~(\ref{oneCol}) above.
\item {\em Every equivalence class of agent $\sq$ contains worlds labelled by all of $\up,\ \down,\ \lef$ and $\righ$, and exactly one of \hearts, \clubs, \diamonds and \spades.}

We define the abbreviation:
\begin{equation}
sq(X) = \know_\sq(\up\lor\down\lor\lef\lor\righ\lor X)\et\susp_\sq\up\et\susp_\sq\down\et\susp_\sq\lef\et\susp_\sq\righ\et\susp_\sq X
\end{equation}
Then we require that every ``square'' satisfies one of this formula for a given card suit:
\begin{equation}\label{oneSuit}
oneSuit = \know_\sq\know_\edge\left[sq(\hearts)\lor sq(\clubs)\lor sq(\diamonds)\lor sq(\spades)\right]
\end{equation}
Combined with the unique labelling specified by (\ref{oneLabel}), this ensures the agent $\sq$ always considers an up-world, a down-world, a left world, a right world and a center world possible. Furthermore any center worlds that the agent considers possible must be labelled by a single suit (\hearts, \clubs, \diamonds\ or \spades).
\item {\em Agent $\edge$ always knows that the world is either: \lef\ or \righ; \up\ or \down; or \hearts, \clubs, \diamonds\ or \spades. However, the agent is not able to further distinguish these worlds.}

This may be specified as follows:
\begin{equation}\label{edge}
edge = \know_\sq\know_\edge\know_\sq\bigvee\!\left[\begin{array}{l}\know_\edge(\lef\lor\righ)\et\susp_\edge\righ\imp\know_\edge(\lef\imp\know_\sq(\up\imp\know_\edge\know_\sq\susp_\edge\righ))\\
						\know_\edge(\up\lor\down)\et\susp_\edge\up\imp\know_\edge(\down\imp\know_\sq(\righ\imp\know_\edge\know_\sq\susp_\edge\up))\\
						\know_\edge(\hearts\lor\clubs\lor\diamonds\lor\spades)\et\susp_\edge\hearts\et\susp_\edge\clubs\et\susp_\edge\diamonds\et\susp_\edge\spades
						\end{array}\right]
\end{equation}
For this equation, we assume that $\know_\sq\know_\edge\know_\sq\phi$ implies that $\phi$ is true at all reachable states. At all such states we can see there are three possibilities for what the agent $\edge$ knows: either $\lef$ or $\righ$ is true, but $\edge$ does not know which (unless it is in the leftmost column of the checkerboard); either $\up$ or \down\ is true, but \edge\ does not know which (unless it is in the bottom row of the checkerboard); or one of \hearts, \clubs, \diamonds\ or \spades\ is true, but \edge\ does not know which. If $\susp_\edge\righ$ is true, then we are in a square that is not in the left column of the board, so the formula $\know_\edge(\lef\imp\know_\sq(\up\imp\know_\edge\know_\sq\susp_\edge\righ))$ ensures that the square above also has a left neighbour. This is sufficient to restrict all squares without a left neighbour to the left column of the board, and a similar case can be made for the bottom row. 

\item {\em A \hearts-square always has a \clubs-square to the right and a \spades-square above; a \clubs-square always has a \hearts-square to the right and a \diamonds-square above; a \diamonds-square always has a \spades-square to the right and a \clubs-square above; and a \spades-square always has a \diamonds-square to the right and a \hearts-square above.} 

We can specify these constraints as:
\begin{eqnarray}
adj_\hearts&=&\know_\sq\know_\edge\left(\hearts\imp\know_\sq\bigwedge\left[\begin{array}{l}\righ\imp\know_\edge(\lef\imp \susp_\sq\clubs)\\\up\imp\know_\edge(\down\imp\susp_\sq\spades)\end{array}\right]\right)\\
adj_\clubs&=&\know_\sq\know_\edge\left(\clubs\imp\know_\sq\bigwedge\left[\begin{array}{l}\righ\imp\know_\edge(\lef\imp \susp_\sq\hearts)\\\up\imp\know_\edge(\down\imp\susp_\sq\diamonds)\end{array}\right]\right)\\
adj_\diamonds&=&\know_\sq\know_\edge\left(\diamonds\imp\know_\sq\bigwedge\left[\begin{array}{l}\righ\imp\know_\edge(\lef\imp \susp_\sq\spades)\\\up\imp\know_\edge(\down\imp\susp_\sq\clubs)\end{array}\right]\right)\\
adj_\spades&=&\know_\sq\know_\edge\left(\spades\imp\know_\sq\bigwedge\left[\begin{array}{l}\righ\imp\know_\edge(\lef\imp \susp_\sq\diamonds)\\\up\imp\know_\edge(\down\imp\susp_\sq\hearts)\end{array}\right]\right)\\
adj &=& adj_\hearts\et adj_\clubs\et adj_\diamonds\et adj_\spades\label{adj}
\end{eqnarray}
Given the properties of the checkerboard model, the formula $\know_\sq\know_\edge\phi$ will ensure that $\phi$ is true at every center world. 
If the proposition \hearts\ is true there, then at every \righ-world in the square, every \lef-world that agent \edge\ cannot distinguish from the \righ-world must belong to the square that is immediately to the right. Therefore the center of that square must be labelled by \clubs. Given the formula (\ref{oneSuit}), it is sufficient for $\susp_\sq\clubs$ to be true at the \lef-world.
Similar reasoning can be given for the other suits and directions.
\end{enumerate}

We define the formula:
\begin{equation}\label{local}
local = oneLabel\et oneSuit \et edge \et adj
\end{equation}
to express the local properties. This leaves us to define the global properties The two global properties we require are:
\begin{itemize}
\item {\em Given any square, the square below the square to the left of the square above the square to its right is identical.}

This is an essential feature of the grid-like construction we require, and enforces that if we go right, up, left, and down, we end up where we began.

\item {\em Every center world of every reachable square is indistinguishable to agent $\edge$.}

By reachable square, we mean any $\sq$-equivalence class that can be reached using the union of the $\edge$ and $\sq$ accessibility relations.
\end{itemize}

While these are the properties we would ideally like, it turns out that they are stronger than what we are able to express. Using Lemma~\ref{nbisimwitness}, the actual properties we will show are:
\begin{itemize}
\item {\em For all $n$, for every center world, $w$, of any square, there is some square below some square to the left of some square above some square to the right of that square, whose center world is $n$-bisimilar to $w$.}

Since this property holds for all $n$ we are able to make an arbitrarily good approximation of the checkerboard configuration.
\item {\em For all $n$, for every center world, $w$, of every \{\edge,\sq\}-reachable square there is some $\edge$-accessible world that is $n$-bisimilar to $w$.}

Therefore, although we may not be able to strictly enforce the global property we desire, we are able to make an arbitrarily good approximation of it.

\end{itemize}

The first property corresponds to the first step of our proof: enforcing a grid like configuration. The second property encodes a common knowledge operator required in step 2. 
To enforce these properties we will require the arbitrary, group or coalition announcements, and the formulation will vary slightly depending on the language. We will therefore address these properties in the context of each variation in the following three subsections.

To give a generalized form of the correctness proof we use a PDL-like notation \cite{fischeretal:1979} to describe composite relations:
\begin{definition}\label{PDL}
A composite program is given by the abstract syntax:
$$\pi ::= a\ |\ A?\ |\ \pi;\pi$$
where $A$ ranges over $\{\hearts,\clubs,\diamonds, \spades, \up, \down, \lef, \righ\}$ and $a$ ranges over $\{\edge,\sq\}$.
Each composite program $\pi$ corresponds to a binary relation $R(\pi)$ on $S$ where:
$$\begin{array}{rcl}
R(A?) &=& \{(s,s)\ |\ S\in V(A)\}\\
R(a) &=& \sim_a\\
R(\pi;\pi') &=& \{(s,t)\ |\ \exists u\ \text{s.t.}\ (s,u)\in R(\pi), (u,t)\in R(\pi')\}
\end{array}$$
\end{definition}
This allows us to discuss a chain of reachable states satisfying a sequence of propositional atoms. For example, if 
we have a world $t$ where $t\in V(\up)$ where $t\sim_\sq u$ and $u\in V(\lef)$ and $u\sim_\edge s$ we can construct a program $\pi = \edge;\lef?;\sq;\up?$, and write $(s,t)\in R(\pi)$. We may also use $s R(\pi)$ to denote the set of all such $t$. However, we refrain from writing $s\sim_\pi t$ to avoid confusion as $\sim_\pi$ would not be an equivalence relation.

As a simplification, we will work with only a finite number of distinct propositions in our model: the labels in $\Lambda$, and the colours in $C$. We let $\Pi = C\cup\Lambda$ and when we refer to bisimilarity, we will mean bisimilarity modulo the atoms in $\Pi$.

\subsection{Arbitrary public announcements}\label{subsec.arbitrary}
The properties are formalized as follows:
\begin{itemize}
\item {\em For all $n$, for every center world, $w$, of any square, there is some square below some square to the left of some square above some square to the right of that square, whose center world is $n$-bisimilar to $w$.}

\begin{eqnarray*}
c_{apa}(X)&=& X\imp\allpub(\know_\sq(\righ\imp(\know_\edge(\lef\imp\know_\sq(\up\imp\know_\edge(\down\imp\\
&&\qquad\know_\sq(\lef\imp\know_\edge(\righ\imp\know_\sq(\down\imp\know_\edge(\up\imp\susp_\sq X))))))))))\\
cyc_{apa} &=& c_{apa}(\hearts)\et c_{apa}(\clubs)\et c_{apa}(\diamonds)\et c_{apa}(\spades)
\end{eqnarray*}

\item {\em For all $n$, for every center world, $w$, of every reachable square there is some $\edge$-accessible world that is $n$-bisimilar to $w$.}

$$\begin{array}{l}
t_{apa}(X, Y, Z) = X\imp\bigwedge\left[\begin{array}{l}\allpub(\know_\edge\neg Y\imp\know_\sq(\righ\imp\know_\edge(\lef\imp\know_\sq\neg Y)))\\
					\allpub(\know_\edge\neg Z\imp\know_\sq(\up\imp\know_\edge(\down\imp\know_\sq \neg Z)))\\
					\allpub(\know_\edge\neg Y\imp\know_\sq(\lef\imp\know_\edge(\righ\imp\know_\sq \neg Y)))\\
					\allpub(\know_\edge\neg Z\imp\know_\sq(\down\imp\know_\edge(\up\imp\know_\sq \neg Z)))
					\end{array}\right]\\
ck_{apa} = t_{apa}(\hearts, \clubs, \spades)\et t_{apa}(\clubs,\hearts,\diamonds)\et t_{apa}(\diamonds,\spades, \clubs)\et t_{apa}(\spades,\diamonds,\hearts)
\end{array}
$$
\end{itemize}
The formula $cyc_{apa}$ enforces the {\em cycle} in a grid: right, up, left and down. The formula $ck_{apa}$ captures the required approximation of {\em common knowledge}.

We define the formula:
\begin{equation}\label{CBA}
CB_{apa} = \know_\edge\know_\sq (local\et cyc_{apa} \et ck_{apa})
\end{equation}

The following lemma provides the main technical result that allows us to define the checkerboard-like model, and hence express the tiling problem.
For brevity, let $card = \hearts\lor\clubs\lor\diamonds\lor\spades$.

\begin{lemma}\label{CB-AA}
Suppose that $M = (S,\sim, V)$, $s\in S$ and $M_s\models CB_{apa}\et card$. 
Let $U = s R(\sq;\righ?;\edge;\lef?;\sq) \cup s R(\sq;\up?;\edge;\down?;\sq) \cup s R(\sq;\lef?;\edge;\righ?;\sq)\cup s R(\sq;\down?;\edge;\up?;\sq)$.
Then:
\begin{enumerate}
\item For all $n\in\N$, for all $t\in U$, if $M_t\models card$, then there is some $v\sim_\edge s$ such that $v\in\nbisim{t}{n}$.
\item For all $n\in\N$, for all $t\in s R(\sq;\righ?;\edge;\lef?;\sq;\up?;\edge;\down?;\sq;\lef?;\edge;\righ?;\sq;\down?;\edge;\up?)$, there is some $v\sim_\sq t$ such that $v\in\nbisim{s}{n}$. 
\end{enumerate}
\end{lemma}

\begin{proof}\quad\\
\begin{enumerate}
\item[1.] 
We will address this case of $t\in s R(\sq;\righ?;\edge;\lef?;\sq)$, and note that the other cases 
  ($t\in s R(\sq;\up?;\edge;\down?;\sq)$, $t\in s R(\sq;\lef?;\edge;\righ?;\sq)$ and $t\in s R(\sq;\down?;\edge;\up?;\sq)$) 
  may be handled in a similar fashion.
Suppose, for contradiction, that there is some $n\in \N$ and some $t\in s R(\sq;\righ?\edge;\lef?;\sq)$ such that $M_t\models card$ and for all $v\sim_\edge s$, $v\notin\nbisim{t}{n}$. Then by Lemma~\ref{nbisimwitness}, there is some epistemic formula $\phi_t^n$ such that for all $w\in S$, $M_w\models\phi_t^n$ if and only if $w\in \nbisim{t}{n}$.
We recall that $card = \hearts\lor\clubs\lor\diamonds\lor\spades$. Now suppose, without loss of generality, $s\in V(\hearts)$.
From the formula $edge$ we know that all $v\sim_\edge s$ must be such that $M_v\models card$, and also that for every $w\in s R(\sq;\edge;\sq)$, $ M_w\models\lef\lor\righ\imp\neg card$.
Therefore, there is some announcement $\psi = \clubs\imp\phi_t^n$. We can observe:
\begin{enumerate}
\item $M_s\models\psi$. Since $M_s\models\hearts$ and $M_s\models oneLabel$, we have $M_s\models\neg\clubs$ so $M_s\models\psi$ since the antecedent is false.
\item For all $v\in s R(\edge;\clubs?)$ we have $M_v\not\models\psi$. As $M_v\models\clubs$ it must satisfy the consequent of $\psi$. However, our assumption is that for all $v\sim_\edge s$, $v\notin\nbisim{t}{n}$, so $M_v\not\models\phi_t^n$.
\item For all $v\in s R(\sq;\righ?)$ we have $M_v\models\psi$. Since $v\in s R(\sq;\righ?)$, we must have $M_v\models\righ$ and from the formula $edge$ we know that $M_v\not\models card$, so $M_v\models\psi$ since the antecedent of $\psi$ is false.
\item For all $v\in s R(\sq;\righ?;\edge;\lef?)$ we have $M_v\models\psi$. As above, since $M_v\models\lef\et edge$ we know that $M_v\not\models card$, so $M_v\models\psi$ since the antecedent of $\psi$ is false.
\item For some $v\in s R(\sq;\righ?;\edge;\lef?;\sq)$ we have $M_v\models\psi\et\clubs$, since $t\in s R(\sq;\righ?;\edge;\lef?;\sq)$, and clearly $M_t\models\clubs\et\phi_t^n$. 
\end{enumerate}
Since $M_s\models ck_{apa}$ we have 
$$M_s\models\hearts\imp\allpub(\know_\edge\neg\clubs\imp\know_\sq(\righ\imp\know_\edge(\lef\imp\know_\sq\neg\clubs)))$$
Substituting $\psi$ as the public announcement we have: 
$$M_s\models[\psi](\know_\edge\neg\clubs\imp\know_\sq(\righ\imp\know_\edge(\lef\imp\know_\sq\neg\clubs)))$$
However, from the reasoning above, in the model $M^\psi = (S^\psi, R^\psi, V^\psi)$ we have 
\begin{enumerate}
\item $s\in S^\psi$, so the announcement $\psi$ is not vacuous.
\item for all $v\in s R(\edge;\clubs?)$, $v\notin S^\psi$ so $M^\psi_s\models\know_\edge\neg\clubs$.
\item for all $v\in s R (\sq;\righ?)$, $v\in S^\psi$.
\item for all $v\in s R (\sq;\righ?;\edge;\lef?)$, $v\in S^\psi$.
\item there is some $t\in s R(\sq;\righ?;\edge;\lef?;\sq;\clubs?)$ where $t\in S^\psi$.
\end{enumerate}
Therefore $M^\psi\models\know_\edge\neg\clubs\et \susp_\sq(\righ\et\susp_\edge(\lef\et\susp_\sq \clubs))$ giving us the required contradiction. 
Hence for all $n\in \N$, for all $t\in s R(\sq;\righ?;\edge;\lef?;\sq)$ where $M_t\models\clubs$, there is some $v\sim_\edge s$ such that $v\in\nbisim{t}{n}$. We may easily generalize the proof to the other cases where $M_s\models\clubs\lor\diamonds\lor\spades$.

    The proofs for the cases: $t\in s R(\sq;\up?;\edge;\down?;\sq)$; $t\in s R(\sq;\lef?;\edge;\righ?;\sq)$; 
    and $t\in s R(\sq;\down?;\edge;\up?;\sq)$) are similar to that of $t\in s R(\sq;\righ?;\edge;\lef?;\sq)$. The only changes are
\begin{enumerate}
\item in the second case, to replace: all references to $\righ$ with $\up$; all references to $\lef$ with $\down$; and the reference to $\clubs$ in the definition of $\psi$ with $\spades$.  
\item in the third case, to replace: all references to $\righ$ with $\lef$; and all references to $\lef$ with $\righ$.  
\item in the fourth case, to replace: all references to $\righ$ with $\down$; all references to $\lef$ with $\up$; and the reference to $\clubs$ in the definition of $\psi$ with $\spades$.
\end{enumerate}

\item[2.] 
For the final part of this proof we are required to show that 
for all $n$, for all $t\in s R(\sq;\righ?;\edge;\lef?;\sq;\up?;\edge;\down?;\sq;\lef?;\edge;\righ?;\sq;\down?;\edge;\up?)$, 
there is some $v\sim_\sq t$ such that $v\in\nbisim{n}{s}$.
We begin by considering any chain of worlds
$$s\sim_\sq t_1\sim_\edge t_2\sim_\sq t_3\sim_\edge t_4\sim_\sq t_5\sim_\edge t_6\sim_\sq t_7\sim_\edge t$$
where $t_1,t_6\in V(\righ)$, $t_2,t_5\in V(\lef)$, $t_3,t\in V(\up)$, and $t_4, t_7\in V(\down)$. 

We first show that from the first four parts of this lemma, each world in this chain satisfies the formula $local$ (\ref{local}).
Because $M_s\models CB_{apa}$, it follows that $M_s\models\know_\edge\know_\sq local$. Therefore $t_1$ must satisfy $local$. 
Furthermore, from the first part of this lemma, we know for all $n\in\N$, for all $t_2\in s R(\sq;\righ?;\edge;\lef?)$, 
for all $w_1\sim_\sq t_2$ where $M_t\models card$, there is some $v_1\sim_\edge s$ such that $w_1\in\nbisim{v_1}{n}$.
Since the modal depth of $local$ is $5$, we may suppose, without loss of generality, that $w_1\in\nbisim{v_1}{14}^\Pi$. 
Since $v_1\sim_\edge s$, we have $M_{v_1}\models CB_{apa}$ and also $M_{v_1}\models\know_\edge\know_\sq local$.
By Lemma~\ref{modDepPres}, $n$-$\Pi$-bisimilarity preserves the interpretation of pure modal formulas with modal depth up to $n$, 
we have $M_{w_1}\models\know_\sq local$. Since $w_1\sim_\sq t_2$ and $t_2\sim_\sq t_3$ 
it follows that $M_{t_2}\models local $ and $M_{t_3}\models local$.
    Since $M_{v_1}\models CB_{apa}\et card$, from the 
    first part of this lemma, given the second component of $U$,
    for all $n\in\N$, for all $v'\in v_1 R(\sq;\up?;\edge;\down?)$, for all $v''\sim_\sq v'$ where $M_{v''}\models card$, 
    there is some $v_2\sim_\edge v_1$ such that $v''\in\nbisim{v_2}{n}$.
As $w_1\in\nbisim{v_1}{14}^\Pi$, for all $w_2\in w_1 R(\sq;\up?;\edge;\down?;\sq)$, such that $M_{w_2}\models card$, 
there is some $v''\in v_1 R(\sq;\up?;\edge;\down?;\sq)$ where $M_{v''}\models card$ and $w_2\in\nbisim{v''}{11}$.
Putting the last two statements together, for all such $w_2$, there is some $v_2\sim_\edge v_1$ such that $w_2\in\nbisim{v_2}{11}$.
As $v_2\sim_\edge v_1$, and $v_1\sim_\edge s$, we have $v_2\sim_\edge s$, and hence $M_{v_2}\models\know_\edge\know_\sq local$. 
As $w_1\sim_\sq t_3$, there is some such $w_2$ where $t_4\sim_\sq w_2$ and $t_5\sim_\sq w_2$.
From Lemma~\ref{modDepPres} we have $M_{t_4}\models local$ and $M_{t_5}\models local$.

    Continuing in this way, (using the 
    first part of the lemma, with the third component of $U$)
    we are able to find some $w_3$ such that $t_6\sim_\sq w_3$, $t_7\sim_\sq w_3$ and for some $v_3\sim_\edge s$, $w_3\in\nbisim{v_3}{8}$. 
    Finally, applying the 
    first part of this lemma with the fourth component of $U$,
    we are able to find $w_4$ such that $t\sim_\sq w_4$, and for some $v_4\sim_\edge s$, $w_4\in\nbisim{v_4}{5}$.
Since the modal depth of $local$ is $5$, it follows that $t_1,\ t_2,\ t_3,\ t_4,\ t_5,\ t_6,\ t_7$ and $t$ all satisfy $local$, so every world has just one label, and every square contains a world labelled by each direction, and so on.

Again, let us suppose that $M_s\models\hearts$, and the proof may be easily adjusted for the other card suits.
The construction we have applied is depicted in Figure~\ref{cbCycles}.
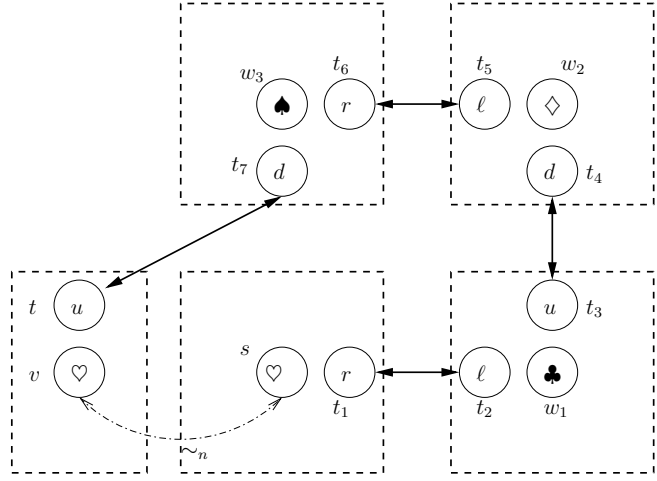
\begin{figure}
\begin{center}
\scalebox{0.7}{
\begin{picture}(0,0)%
\includegraphics{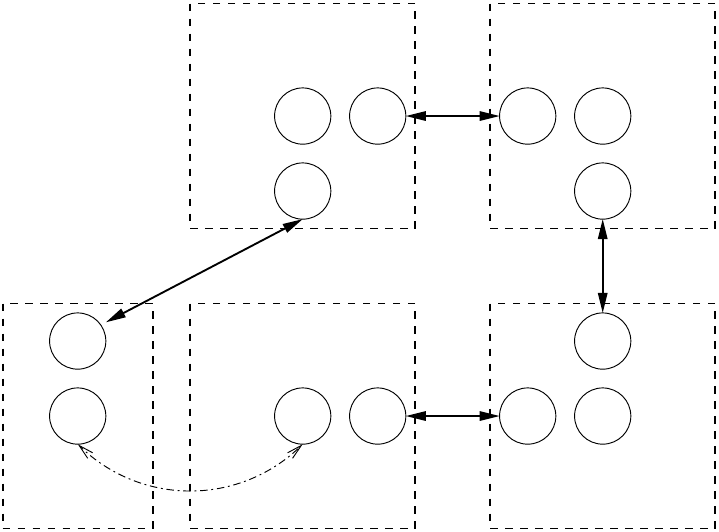}%
\end{picture}%
\setlength{\unitlength}{3947sp}%
\begingroup\makeatletter\ifx\SetFigFont\undefined%
\gdef\SetFigFont#1#2#3#4#5{%
  \reset@font\fontsize{#1}{#2pt}%
  \fontfamily{#3}\fontseries{#4}\fontshape{#5}%
  \selectfont}%
\fi\endgroup%
\begin{picture}(5744,4244)(-321,-8183)
\put(1201,-8011){\makebox(0,0)[lb]{\smash{{\SetFigFont{12}{14.4}{\rmdefault}{\mddefault}{\updefault}{\color[rgb]{0,0,0}$\sim_n$}%
}}}}
\put(1951,-7336){\makebox(0,0)[lb]{\smash{{\SetFigFont{12}{14.4}{\rmdefault}{\mddefault}{\updefault}{\color[rgb]{0,0,0}$\hearts$}%
}}}}
\put(2626,-7336){\makebox(0,0)[lb]{\smash{{\SetFigFont{12}{14.4}{\rmdefault}{\mddefault}{\updefault}{\color[rgb]{0,0,0}$\righ$}%
}}}}
\put(2551,-7636){\makebox(0,0)[lb]{\smash{{\SetFigFont{12}{14.4}{\rmdefault}{\mddefault}{\updefault}{\color[rgb]{0,0,0}$t_1$}%
}}}}
\put(3826,-7336){\makebox(0,0)[lb]{\smash{{\SetFigFont{12}{14.4}{\rmdefault}{\mddefault}{\updefault}{\color[rgb]{0,0,0}$\lef$}%
}}}}
\put(3826,-7636){\makebox(0,0)[lb]{\smash{{\SetFigFont{12}{14.4}{\rmdefault}{\mddefault}{\updefault}{\color[rgb]{0,0,0}$t_2$}%
}}}}
\put(4426,-7636){\makebox(0,0)[lb]{\smash{{\SetFigFont{12}{14.4}{\rmdefault}{\mddefault}{\updefault}{\color[rgb]{0,0,0}$w_1$}%
}}}}
\put(4426,-7336){\makebox(0,0)[lb]{\smash{{\SetFigFont{12}{14.4}{\rmdefault}{\mddefault}{\updefault}{\color[rgb]{0,0,0}$\clubs$}%
}}}}
\put(4426,-6736){\makebox(0,0)[lb]{\smash{{\SetFigFont{12}{14.4}{\rmdefault}{\mddefault}{\updefault}{\color[rgb]{0,0,0}$\up$}%
}}}}
\put(4426,-5536){\makebox(0,0)[lb]{\smash{{\SetFigFont{12}{14.4}{\rmdefault}{\mddefault}{\updefault}{\color[rgb]{0,0,0}$\down$}%
}}}}
\put(4426,-4936){\makebox(0,0)[lb]{\smash{{\SetFigFont{12}{14.4}{\rmdefault}{\mddefault}{\updefault}{\color[rgb]{0,0,0}$\diamonds$}%
}}}}
\put(3826,-4936){\makebox(0,0)[lb]{\smash{{\SetFigFont{12}{14.4}{\rmdefault}{\mddefault}{\updefault}{\color[rgb]{0,0,0}$\lef$}%
}}}}
\put(2626,-4936){\makebox(0,0)[lb]{\smash{{\SetFigFont{12}{14.4}{\rmdefault}{\mddefault}{\updefault}{\color[rgb]{0,0,0}$\righ$}%
}}}}
\put(2026,-4936){\makebox(0,0)[lb]{\smash{{\SetFigFont{12}{14.4}{\rmdefault}{\mddefault}{\updefault}{\color[rgb]{0,0,0}$\spades$}%
}}}}
\put(2026,-5536){\makebox(0,0)[lb]{\smash{{\SetFigFont{12}{14.4}{\rmdefault}{\mddefault}{\updefault}{\color[rgb]{0,0,0}$\down$}%
}}}}
\put(1726,-4636){\makebox(0,0)[lb]{\smash{{\SetFigFont{12}{14.4}{\rmdefault}{\mddefault}{\updefault}{\color[rgb]{0,0,0}$w_3$}%
}}}}
\put(1651,-5461){\makebox(0,0)[lb]{\smash{{\SetFigFont{12}{14.4}{\rmdefault}{\mddefault}{\updefault}{\color[rgb]{0,0,0}$t_7$}%
}}}}
\put(-149,-6736){\makebox(0,0)[lb]{\smash{{\SetFigFont{12}{14.4}{\rmdefault}{\mddefault}{\updefault}{\color[rgb]{0,0,0}$t$}%
}}}}
\put(226,-6736){\makebox(0,0)[lb]{\smash{{\SetFigFont{12}{14.4}{\rmdefault}{\mddefault}{\updefault}{\color[rgb]{0,0,0}$\up$}%
}}}}
\put(226,-7336){\makebox(0,0)[lb]{\smash{{\SetFigFont{12}{14.4}{\rmdefault}{\mddefault}{\updefault}{\color[rgb]{0,0,0}$\hearts$}%
}}}}
\put(-149,-7336){\makebox(0,0)[lb]{\smash{{\SetFigFont{12}{14.4}{\rmdefault}{\mddefault}{\updefault}{\color[rgb]{0,0,0}$v$}%
}}}}
\put(2551,-4561){\makebox(0,0)[lb]{\smash{{\SetFigFont{12}{14.4}{\rmdefault}{\mddefault}{\updefault}{\color[rgb]{0,0,0}$t_6$}%
}}}}
\put(3826,-4561){\makebox(0,0)[lb]{\smash{{\SetFigFont{12}{14.4}{\rmdefault}{\mddefault}{\updefault}{\color[rgb]{0,0,0}$t_5$}%
}}}}
\put(4576,-4561){\makebox(0,0)[lb]{\smash{{\SetFigFont{12}{14.4}{\rmdefault}{\mddefault}{\updefault}{\color[rgb]{0,0,0}$w_2$}%
}}}}
\put(4801,-5536){\makebox(0,0)[lb]{\smash{{\SetFigFont{12}{14.4}{\rmdefault}{\mddefault}{\updefault}{\color[rgb]{0,0,0}$t_4$}%
}}}}
\put(4801,-6736){\makebox(0,0)[lb]{\smash{{\SetFigFont{12}{14.4}{\rmdefault}{\mddefault}{\updefault}{\color[rgb]{0,0,0}$t_3$}%
}}}}
\put(1726,-7111){\makebox(0,0)[lb]{\smash{{\SetFigFont{12}{14.4}{\rmdefault}{\mddefault}{\updefault}{\color[rgb]{0,0,0}$s$}%
}}}}
\end{picture}%
}
\end{center}
\caption{A representation of the final case of Lemma~\ref{CB-AA}. The dashed boxes represent the $\sq$ equivalence classes; the solid lines are the $\edge$ equivalence classes; and the dashed line indicates $n$-$\Pi$-bisimilarity.}\label{cbCycles}
\end{figure}
We are required to show for every such $t$, for any given $n$, there is some $v\sim_\sq t$ such that $v\in \nbisim{s}{n}$. Suppose that this was not the case. 
Then for some $n\in \N$, for some 
$$t\in s R(\sq;\righ?;\edge;\lef?;\sq;\up?;\edge;\down?;\sq;\lef?;\edge;\righ?;\sq;\down?;\edge;\up?),$$ 
we have for all $v\sim_\sq t$, $v\notin\nbisim{s}{n}^\Pi$.
From Lemma~\ref{nbisimwitness} for every $n$, there is some formula, $\phi^n_s$, such that for all $w\in S$, $M_w\models\phi^n_s$ if and only if $w\in\nbisim{s}{n}$.
Let $\psi$ be the formula $\hearts\imp\phi^n_s$. Then:
\begin{enumerate}
\item $M_s\models\psi$. It is clear that $s$ is $n$-$\Pi$-bisimilar to itself, so $M_s\models\phi^n_s$.
\item $M_{t_i}\models\psi$, for $i = 1,\hdots,7$. Since $M_{t_i}\models local$, we have $M_{t_i}\models oneLabel$. We also have $M_{t_i}\models \lef\lor\righ\lor\up\lor\down$, so $M_{t_i}\not\models\hearts$. As the antecedent is not satisfied, we have $M_{t_i}\models\psi$
\item $M_t\models\psi$. This follows from the reasoning used in the previous case, since we have $M_t\models local$ and $M_t\models\up$.
\end{enumerate}
Therefore after announcing $\psi$, the worlds $s, t_1,\hdots, t_7$, and $t$ would all remain in the model $M^\psi = (S^\psi, R^\psi, V^\psi)$. 
However, we know that there is no $v\sim_\sq t$ such that $v\in\nbisim{s}{n}$, so for all $v\in t R_\sq$, we have $M_{v}\models\neg\phi^n_s$. 
For all $v\in t R_\sq\cap S^\psi$ we must have $M_{v}\models(\hearts\imp\phi)\et\neg\phi^n_s$, so it must be that $M_{v}\models\neg\hearts$.
As propositions are preserved by announcements, we have $M^\psi_t\models\know_\sq\neg\hearts$, and hence
\begin{equation}
M^\psi_s\models\susp_\sq(\righ\et(\susp_\edge(\lef\et\susp_\sq(\up\et\susp_\edge(\down\et\susp_\sq(\lef\et\susp_\edge(\righ\et\susp_\sq(\down\et\susp_\edge(\up\et\know_\sq\neg\hearts)\!)\!)\!)\!)\!)\!)\!)\!)
\end{equation}
It follows that $M_s\models\neg c_{apa}(\hearts)$ giving us the required contradiction. A similar argument can be given for each of the other card suits, completing the proof.
\end{enumerate}
\end{proof}

\subsection{Group announcements}\label{subsec.group} 

The approach for group announcements is similar to to the approach for arbitrary public announcements, except we are now able to use group announcements, which allow for a more direct proof. 
We note that the approach for coalition announcements below could also be applied to group announcement logic, since it only uses the coalition of both agents, which is equivalent to a group announcement from both agents. 
However, we have chosen to present a separate proof for group announcement logic since it only requires us to use single agent groups, and is consequently a stronger result. 
The properties are formalized as follows:
\begin{itemize}
\item {\em For all $n$, for every center world, $w$, of any square, there is some square below some square to the left of some square above some square to the right of that square, whose center world is $n$-bisimilar to $w$.}

\begin{eqnarray*}
c_{ga}(X)&=& X\imp\allgrp{\sq}(\know_\sq(\righ\imp(\know_\edge(\lef\imp\know_\sq(\up\imp\know_\edge(\down\imp\\
  &&\qquad\know_\sq(\lef\imp\know_\edge(\righ\imp\know_\sq(\down\imp\susp_\edge(\up\land\susp_\sq X))))))))))\\
cyc_{ga} &=& c_{ga}(\hearts)\et c_{ga}(\clubs)\et c_{ga}(\diamonds)\et c_{ga}(\spades)
\end{eqnarray*}

\item {\em For all $n$, for every center world, $w$, of every reachable square there is some $\edge$-accessible world that is $n$-bisimilar to $w$.}
$$
\begin{array}{l}
t_{ga}(X, Y, Z) = X\imp\allgrp{\edge}\know_\sq\bigwedge\left[\begin{array}{l}\righ\imp\know_\edge(\lef\imp\susp_\sq Y)\\
									      \up\imp\know_\edge(\down\imp\susp_\sq Z)\\
									      \lef\imp\know_\edge(\righ\imp\susp_\sq Y)\\
									      \down\imp\know_\edge(\up\imp\susp_\sq Z)
										\end{array}\right]\\
ck_{ga} = t_{ga}(\hearts, \clubs, \spades)\et t_{ga}(\clubs,\hearts,\diamonds)\et t_{ga}(\diamonds,\spades, \clubs)\et t_{ga}(\spades,\diamonds,\hearts)
\end{array}$$
\end{itemize}

We define the formula:
\begin{equation}\label{CBG}
CB_{ga} = \know_\sq\know_\edge\know_\sq (local\et cyc_{ga} \et ck_{ga})
\end{equation}

It is worth taking the time to compare the formulas $cyc_{ga}$ and $ck_{ga}$ to the corresponding formulations for arbitrary announcements: respectively $cyc_{apa}$ and $ck_{apa}$. Consider a simple fragment of these:
\begin{eqnarray*}
ck_{ga}^* &=& \hearts\imp\allgrp{\edge}\know_\sq\left[\righ\imp\know_\edge(\lef\imp\susp_\sq \clubs)\right]\\
ck_{apa}^* &=& \hearts\imp\allpub\left[\know_\edge\neg \clubs\imp\know_\sq(\righ\imp\know_\edge(\lef\imp\know_\sq\neg \clubs))\right]
\end{eqnarray*}
While each formula has a similar purpose and structure, the formulas used for group announcements are more direct.
The purpose of $ck_{ga}^*$ is to say that agent $\edge$ cannot distinguish the center of a $\hearts$-square from the center of the $\clubs$-square to its right. As any group announcement made by $\edge$ at the center of the hearts square cannot rule out any world that agent $\edge$ cannot distinguish from the center of the hearts world, the formula simply needs to assert that after any such group announcement, if the intermittent edge worlds remain (the $\righ$ and $\lef$ worlds), then the center of the $\clubs$ square must also remain. However, achieving the same using arbitrary announcements is more complicated. After making any public announcement, if there are no $\clubs$-worlds  accessible by agent $\edge$ from the center of the $\hearts$-square, then there can be no $\clubs$ at the center of the square to the right. The correctness of this construction is shown in Lemma~\ref{CB-AA}. 

Essentially, in GAL we are quantifying over what an agent {\em knows} to be true, so the accessibility relation for that agent is implicit in the quantification whereas, in APAL we quantify over all formulas, so we must make the reachability explicit in the formula. The following lemma has a similar form to Lemma~\ref{CB-AA} and provides the main technical result that allows us to define the checkerboard like model, and hence express the tiling problem. Note that whilst the Lemma below is very similar to the statement of Lemma~\ref{CB-AA}, there are some small but significant differences between the two.

\begin{lemma}\label{CB-GA}
Suppose that $M = (S,\sim, V)$, $s\in S$ and $M_s\models CB_{ga}\et card$. 
Let $U = s R(\sq;\righ?;\edge;\lef?;\sq) \cup s R(\sq;\up?;\edge;\down?;\sq) \cup s R(\sq;\lef?;\edge;\righ?;\sq)\cup s R(\sq;\down?;\edge;\up?;\sq)$.
Then:
\begin{enumerate}
\item For all $n\in\N$, for all $t'\in U$, there is some $t\sim_\sq t'$ and some $v\sim_\edge s$ such that $v\in\nbisim{t}{n}$.
\item For all $n\in\N$, for all 
$t \in s R(\sq;\righ?;\edge;\lef?;\sq;\up?;\edge;\down?;\sq;\lef?;\edge;\righ?;\sq;\down?)$, there is some $v\sim_\edge t$ and some $w\sim_\sq v$ such that $M_v\models \up$ and $w\in\nbisim{s}{n}$.
\end{enumerate}
\end{lemma}

\begin{proof}\quad\\
\begin{enumerate}
\item[1.] 
We will assume $t'\in s R(\sq;\righ?;\edge;\lef?;\sq)$ and the cases for $t'\in s R(\sq;\up?;\edge;\down?;\sq)$, $t'\in s R(\sq;\lef?;\edge;\righ?;\sq)$
and $t'\in s R(\sq;\down?;\edge;\up?;\sq)$ are similar.
Suppose, for contradiction, that there is some $n\in \N$ and some $t'\in s R(\sq;\righ?\edge;\lef?)$ such that for all $t\sim_\sq t'$, 
for all $v\sim_\edge s$, $v\notin\nbisim{t}{n}$. Then, by Lemma~\ref{finitenbisim}, there is a set of epistemic formulas $\phi_0,\hdots,\phi_m$ such that:
\begin{enumerate}
\item for all $v\sim_\edge s$ there is some $a\leq m$ such that $M_v\models\phi_a$;
\item for all $a\leq m$, there is some $v\sim_\edge s$ such that $M_v\models\phi_a$; and
\item for all $a\leq m$, for all $v\sim_\edge s$ where $M_v\models\phi_a$, for all $t\in S$, $M_t\models\phi_a$ if and only if $t\in\nbisim{v}{n}$.
\end{enumerate} 
It follows that $M_s\models\know_\edge\bigvee_{a\leq m}\phi_a$, and by our assumption, for all $t\sim_\sq t'$ we have $M_t\models\neg\know_\edge\bigvee_{a\leq m}\phi_a$.
    Let $\psi$ be the formula 
    $\know_\edge(\lef\lor\righ\lor\bigvee_{a\leq m}\phi_a)$.
    (Even though, from the fact $M_s\models local$, we have $M_s\models\know_\edge(\neg\lef\et\neg\righ)$, 
    we include the propositions $\lef$ and $\righ$ in the disjunction for $\psi$ since they do not stop agent $\edge$ knowing $\psi$ is true, 
    but they do prevent the intermittent $\lef$ and $\righ$ worlds from being removed by the announcement of $\psi$.)

In the model $M^\psi = (S^\psi, R^\psi, V^\psi)$ we have 
\begin{enumerate}
\item $s\in S^\psi$, since for some $a\leq m$, $M_s\models\phi_a$, so the announcement $\psi$ is not vacuous.
\item for all $t'\in s R (\sq;\righ?)$, $t'\in S^\psi$, since $M_{t'}\models\know_\edge\righ$
\item for all $t'\in s R (\sq;\righ?;\edge;\lef?)$, $t'\in S^\psi$, since $M_{t'}\models\know_\edge\lef$
\item for some $t'\in s R(\sq;\righ?;\edge;\lef?)$ for every $t\sim_\sq t'$, if $t\in S^\psi$, then $M_t\models\neg card$. 
  By our assumption that for all $t \sim_\sq t'$, for all $v\sim_\edge s$, $t\notin\nbisim{v}{n}$, we have $M_t\not\models\bigvee_{a\leq m}\phi_a$. 
    Since $M_s\models\know_\sq\know_\edge local$, it follows that 
    $M_t\models card \imp\know_\edge(\neg\lef\et\neg\righ)$, so $M_t\models card\imp\neg\psi$.
\end{enumerate}
Therefore $M^\psi_s\models\susp_\sq(\righ\et\susp_\edge(\lef\et\know_\sq\neg card))$. 
Hence $M_s\models card\imp\langle\psi\rangle\susp_\sq(\righ\et\susp_\edge(\lef\et\know_\sq\neg card))$ which is sufficient to contradict that $M_s\models ck_{ga}$. 
Note that regardless of how $X$, $Y$, and $Z$ are applied in the definition of $t_{ga}(X,Y,Z)$ this proof only requires that $(X\lor Y\lor Z)\imp card$.

As in Lemma~\ref{CB-AA}, the cases for
$t'\in s R(\sq;\up?;\edge;\down?;\sq)$, $t'\in s R(\sq;\lef?;\edge;\righ?;\sq)$
and $t'\in s R(\sq;\down?;\edge;\up?;\sq)$ and accounted for in a similar manner.
\item[2.] For the final part of this proof we are required to show that 
  for all $n$, for all 
  $t\in s R(\sq;\righ?;\edge;\lef?;\sq;\up?;\edge;\down?;\sq;\lef?;\edge;\righ?;\sq;\down?)$,
    there is some $v\sim_\edge t$ and $w\sim_\sq v$ such that $M_v\models v$ and $w\in\nbisim{n}{s}$.
    We begin by considering any chain of worlds
    $$s\sim_\sq t_1\sim_\edge t_2\sim_\sq t_3\sim_\edge t_4\sim_\sq t_5\sim_\edge t_6\sim_\sq t$$
    where $t_1,t_6\in V(\righ)$, $t_2,t_5\in V(\lef)$, $t_3\in V(\up)$, and $t_4, t\in V(\down)$. 

We first show that from the first part of this lemma, each world in this chain satisfies the formula $local$ (\ref{local}).
    We will suppose without loss of generality, that $M_s\models\hearts$.
Because $M_s\models CB_{ga}$, it follows that $M_s\models\know_\sq\know_\edge local$. Therefore $t_1$ must satisfy $local$. 
    Furthermore, from the first part of this lemma 
    with the first component from $U$,
    we know for all $n\in\N$, 
    for all $t_2\in s R(\sq;\righ?;\edge;\lef?)$, there is some $w_1\sim_\sq t_2$ where $M_{w_1}\models card$, and some 
    $v_1\sim_\edge s$ such that $w_1\in\nbisim{v_1}{n}$.
Since the modal depth of $local$ is $5$, we may suppose, without loss of generality, that $w_1\in\nbisim{v_1}{14}^\Pi$. 
Since $v_1\sim_\edge s$, we have $M_{v_1}\models CB_{ga}$ and also $M_{v_1}\models\know_\edge\know_\sq local$.
By Lemma~\ref{modDepPres}, $n$-$\Pi$-bisimilarity preserves the interpretation of pure modal formulas with modal depth up to $n$, we have $M_{w_1}\models\know_\sq local$. 
Since $w_1\sim_\sq t_2$ and $t_2\sim_\sq t_3$ it follows that $M_{t_2}\models local$ and $M_{t_3}\models local$.
    Since $M_{v_1}\models CB_{ga}\et card$, from the 
    first part of this lemma with the second component of $U$,
    for all $n\in\N$, for all $v'\in v_1 R(\sq;\up?;\edge;\down?)$, for all $v''\sim_\sq v'$ where $M_{v''}\models card$, 
    there is some $v_2\sim_\edge v_1$ such that $v''\in\nbisim{v_2}{n}$.
As $w_1\in\nbisim{v_1}{14}^\Pi$, for every  $t'\in w_1 R(\sq;\up?;\edge;\down?)$, there is some $w_2\sim_\sq t'$ such that $M_{w_2}\models card$, 
and some $v''\in v_1 R(\sq;\up?;\edge;\down?;\sq)$ where $M_{v''}\models card$ and $w_2\in\nbisim{v''}{11}$.
Putting the last two statements together, for every $t'\in w_1 R(\sq;\up?;\edge;\down?)$, there is some $w_2\sim_\sq t'$ and some $v_2\sim_\edge v_1$ such that $w_2\in\nbisim{v_2}{11}$.
As $v_2\sim_\edge v_1$, and $v_1\sim_\edge s$, we have $v_2\sim_\edge s$, and hence $M_{v_2}\models\know_\edge\know_\sq local$. 
As $w_1\sim_\sq t_3$, there is some such $w_2$ and $v_2$ such that $t_4\sim_\sq w_2$ and $t_5\sim_\sq w_2$.
From Lemma~\ref{modDepPres} we have $M_{t_4}\models local$ and $M_{t_5}\models local$.

    Continuing in this way, 
    (using the third component of $U$)
    we are able to find some $w_3$ such that $t_6\sim_\sq w_3$, 
    $t\sim_\sq w_3$
    and for some $v_3\sim_\edge s$, $w_3\in\nbisim{v_3}{8}$. 
    Since the modal depth of $local$ is $5$, it follows that $t_1,\ t_2,\ t_3,\ t_4,\ t_5,\ t_6$ and $t$ all satisfy $local$, 
    so every world has just one label, and every square contains a world labelled by each direction, and so on.

    We are required to show for every such $t$, for any given $n$, there is some $v\sim_\edge t$ 
    and some $w\sim_\sq v$ such that $M_v\models u$ and $w\in \nbisim{s}{n}$.
    Suppose that this was not the case. 
    Then for some $n\in \N$, for some 
    $t\in s R(\sq;\righ?;\edge;\lef?;\sq;\up?;\edge;\down?;\sq;\lef?;\edge;\righ?;\sq;\down?)$, 
    we have for all $v\sim_\edge t$ where $M_v\models \up$, for all $w\sim_\sq v$, $w\notin\nbisim{s}{n}$.
    From Lemma~\ref{finitenbisim} for every $n$, there is a set of formulas, $\phi_0,\hdots,\phi_m$, such that 
    \begin{enumerate}
      \item for all $s'\sim_\sq s$ there is some $a\leq m$ such that $M_{s'}\models\phi_a$, 
        and for all $t'\in S$, $M_{t'}\models\phi_a$ if and only if $t'\in\nbisim{s'}{n}$.
      \item for all $a\leq m$ there is some $s'\sim_\sq s$ such that $M_{s'}\models\phi_a$;
    \end{enumerate}
    Let $\psi$ be the formula 
    $\know_\sq(\hearts\imp\know_\sq\bigvee_{a\leq m}\phi_a)$.
    Then:
    \begin{enumerate}
      \item $M_s\models\psi$. It is clear that for every  $s'\sim_\sq s$, $s'$ is $n$-$\Pi$-bisimilar to itself, 
        so $M_{s'}\models\bigvee_{a\leq m}\phi_a$. Therefore $M_s\models\know_\sq\psi$.
      \item $M_{t_i}\models\psi$, for $i = 1,\hdots,6$. Since $M_{t_i}\models local$, 
        we have $M_{t_i}\models oneLabel$. 
        We also have 
        $M_{t_i}\not\models \know_\sq\hearts$. 
        As the antecedent is not satisfied, we have $M_{t_i}\models\psi$
      \item $M_t\models\psi$. This follows from the reasoning used in the previous case, since we have $M_t\models local$ and 
        $M_t\models\know_\sq\neg\hearts$.
    \end{enumerate}
    Therefore, at $s$ agent $\sq$ may make the announcement $\psi$, and after announcing $\psi$, the worlds 
    $s, t_1,\hdots, t_6$,
    and $t$ would all remain in the model $M^\psi = (S^\psi, R^\psi, V^\psi)$. 
    However, we know that there is no $v\sim_\edge t$ where $M_v\models\up$, 
    such that for some $w\sim_\sq v$, $w\in\nbisim{s}{n}$. 
    So for all $v\in t R_\edge$, we have 
    $M_v\models\know_\sq(\up\imp\know_\sq\neg\bigvee_{a\leq m}\phi_a)$.
    Therefore, for all $v\in t R_\edge\cap S^\psi$ where $M_v\models\up$, for all $w\in v R_\sq\cap S^\psi$,
    we must have $M_w\models\know_\sq\neg\hearts$.
    As propositions are preserved by announcements, we have 
    $M^\psi_t\models\know_\sq\neg\hearts$, 
    and hence
    \begin{eqnarray*}
      M^\psi_s&\models&\susp_\sq(\righ\et(\susp_\edge(\lef\et\susp_\sq(\up\et\susp_\edge(\down\et\\
      &&\qquad\susp_\sq(\lef\et\susp_\edge(\righ\et\susp_\sq(\down\et\know_\edge(\up\imp\know_\sq\neg\hearts)))))))))
    \end{eqnarray*}
    It follows that $M_s\models\neg cyc_{ga}$ (regardless of suit) giving us the required contradiction.
\end{enumerate}
\end{proof}

\subsection{Coalition announcements}\label{subsec.coalition}

The proof for coalition announcement logic is a combination of the two previous cases. We use a coalition of both agents (which is effectively a group announcement), and the constructions for the arbitrary public announcement logic. 
We briefly compare the constructions below to the previous cases for group announcements and arbitrary announcements. 
With coalition announcements we can no longer rely on the accessibility relations being implicit in the quantified announcement, as with group announcements, since the quantification is now over the accessibility relations of all agents.
As a result we must use the more complicated constructions of APAL, along with the proof methods of GAL. 
The properties are formalized as follows:
\begin{itemize}
\item {\em For all $n$, for every center world, $w$, of any square, there is some square below some square to the left of some square above some square to the right of that square, whose center world is $n$-bisimilar to $w$.}

\begin{eqnarray*}
c_{ca}(X)&=& X\imp\allcoal{\sq,\edge} (\know_\sq(\righ\imp(\know_\edge(\lef\imp\know_\sq(\up\imp\know_\edge(\down\imp\\
					&&\qquad\know_\sq(\lef\imp\know_\edge(\righ\imp\know_\sq(\down\imp\know_\edge(\up\imp\susp_\sq X))))))))))\\
cyc_{ca} &=& c_{ca}(\hearts)\et c_{ca}(\clubs)\et c_{ca}(\diamonds)\et c_{ca}(\spades)
\end{eqnarray*}

\item {\em For all $n$, for every center world, $w$, of every reachable square there is some $\edge$-accessible world that is $n$-bisimilar to $w$.}

$$\begin{array}{l}
t_{ca}(X, Y, Z) = X\imp\bigwedge\left[\begin{array}{l}\allcoal{\sq,\edge}(\know_\edge\neg Y\imp\know_\sq(\righ\imp\know_\edge(\lef\imp\know_\sq\neg Y)))\\
					\allcoal{\sq,\edge}(\know_\edge\neg Z\imp\know_\sq(\up\imp\know_\edge(\down\imp\know_\sq\neg Z)))\\
					\allcoal{\sq,\edge}(\know_\edge\neg Y\imp\know_\sq(\lef\imp\know_\edge(\righ\imp\know_\sq\neg Y)))\\
					\allcoal{\sq,\edge}(\know_\edge\neg Z\imp\know_\sq(\down\imp\know_\edge(\up\imp\know_\sq\neg Z)))
					\end{array}\right]\\
ck_{ca} = t_{ca}(\hearts, \clubs, \spades)\et t_{ca}(\clubs,\hearts,\diamonds)\et t_{ca}(\diamonds,\spades, \clubs)\et t_{ca}(\spades,\diamonds,\hearts)
\end{array}
$$
\end{itemize}

We define the formula:
\begin{equation}\label{CBC}
CB_{ca} = \know_\sq\know_\edge\know_\sq (local\et cyc_{ca} \et ck_{ca})
\end{equation}

The following lemma provides the main technical result that allows us to define the checkerboard like model, and hence express the tiling problem.

\begin{lemma}\label{CB-CA}
Suppose that $M = (S,\sim, V)$, $s\in S$ and $M_s\models CB_{ca}\et card$. 
Let $U = s R(\sq;\righ?;\edge;\lef?;\sq) \cup s R(\sq;\up?;\edge;\down?;\sq) \cup s R(\sq;\lef?;\edge;\righ?;\sq)\cup s R(\sq;\down?;\edge;\up?;\sq)$.
Then:
\begin{enumerate}
\item For all $n\in\N$, for all $t\in U$, if $M_t\models card$, then there is some $v\sim_\edge s$ such that $v\in\nbisim{t}{n}$.
\item For all $n\in\N$, for all $t \in s R(\sq;\righ?;\edge;\lef?;\sq;\up?;\edge;\down?;\sq;\lef?;\edge;\righ?;\sq;\down?;\edge;\up?)$, there is some $v\sim_\sq t$ such that $v\in\nbisim{s}{n}$. 
\end{enumerate}
\end{lemma}

\begin{proof}\quad\\
\begin{enumerate}
\item[1.] 
As in the proof of Lemma~\ref{CB-AA}, we address the case for $t\in s R(\sq;\righ?;\edge;\lef?;\sq)$ 
and the cases of $t\in s R(\sq;\up?;\edge;\down?;\sq)$, $t\in s R(\sq;\lef?;\edge;\righ?;\sq)$ and $t\in s R(\sq;\down?;\edge;\up?;\sq)$ are done in a similar manner.
Also suppose, without loss of generality, that $M_s\models\hearts$. The following proof may be easily adapted to other suits.
Now suppose for contradiction that there is some $n\in \N$ and some $t\in s R(\sq;\righ?\edge;\lef?;\sq)$ where $M_t\models card$, and for all $v\sim_\edge s$, $v\notin\nbisim{t}{n}$. 
Then, by Lemma~\ref{nbisimwitness}, there is some epistemic formula $\phi_t^n$ such that for all $t'\in S$, $M_{t'}\models\phi_t^n$ if and only if $t'\in\nbisim{t}{n}$.
As for all $v\sim_\edge s$, we have $v\notin\nbisim{t}{n}$, it follows that$M_v\models\neg\phi_t^n$. Therefore $M_v\models\clubs\imp\neg\phi_t^n$ and thus $M_s\models\know_\edge(\clubs\imp\neg\phi_t^n$). 
(Introducing the antecedent $\clubs$ prevents the announcement of $\know_\edge\phi_t^n$ from removing the intermittent $\lef$ and $\righ$ worlds that we require to construct our contradiction.) 
Because $M_s\models local$, we have $M_s\models adj_\hearts$ so $M_t\models\clubs$. Therefore, in the model $M^\psi = (S^\psi, R^\psi, V^\psi)$ we have 
\begin{enumerate}
\item $s\in S^\psi$, since $M_s\models\hearts$ so $M_s\models\neg\clubs$ and the antecedent of $\psi$ is not satisfied.
\item for all $t'\in s R (\sq;\righ?)$, $t'\in S^\psi$, since $M_{t'}\models\righ$ so $M_{t'}\models\neg\clubs$ and the antecedent of $\psi$ is not satisfied.
\item for all $t'\in s R (\sq;\righ?;\edge;\lef?)$, $t'\in S^\psi$, since $M_{t'}\models\lef$ so the antecedent of $\psi$ is not satisfied.
\item for all $t'\sim_\edge s$, where $M_{t'}\models\clubs$, we have $t'\notin S^\psi$, since $t'\notin\nbisim{t}{n}$, so $M_{t'}\models\neg\phi_t^n$.
\item $t\in S^\psi$ since clearly $t\in\nbisim{t}{n}$ so $M_t\models\phi_t^n$.
\end{enumerate}
Therefore $M^\psi_s\models\know_\edge\neg\clubs\et\susp_\sq(\righ\et\susp_\edge(\lef\et\susp_\sq \clubs))$. 
As the coalition announcement may consist of $\know_\edge\psi$ and $\know_\sq \top$ we have 
$$M_s\models\hearts\et\somecoal{\sq,\edge}(\know_\edge\neg\clubs\et\susp_\sq(\righ\et\susp_\edge(\lef\et\susp_\sq \clubs)))$$
which is sufficient to contradict our assumption $M_s\models CB_{ca}\et\hearts$.

As before, the proofs for $t\in s R(\sq;\up?;\edge;\down?;\sq)$, $t\in s R(\sq;\lef?;\edge;\righ?;\sq)$ and $t\in s R(\sq;\down?;\edge;\up?;\sq)$ are done similarly.

\item[2.] For the final part of this proof we are required to show that 
for all $n$, for all $t\in s R(\sq;\righ?;\edge;\lef?;\sq;\up?;\edge;\down?;\sq;\lef?;\edge;\righ?;\sq;\down?;\edge;\up?)$, there is some $v\sim_\sq t$ such that $v\in\nbisim{n}{s}$.
We begin by considering any chain of worlds
$$s\sim_\sq t_1\sim_\edge t_2\sim_\sq t_3\sim_\edge t_4\sim_\sq t_5\sim_\edge t_6\sim_\sq t_7\sim_\edge t$$
where $t_1,t_6\in V(\righ)$, $t_2,t_5\in V(\lef)$, $t_3,t\in V(\up)$, and $t_4, t_7\in V(\down)$. 

We may show that from the first four parts of this lemma, each world in this chain satisfies the formula $local$ (\ref{local}), 
using the same argument from the proof of the fifth clause of Lemma~\ref{CB-AA}.
We will assume that $M_s\models\hearts$, and cases for other suits have similar proofs.
We are required to show for every such $t\in s R(\sq;\righ?;\edge;\lef?;\sq;\up?;\edge;\down?;\sq;\lef?;\edge;\righ?;\sq;\down?;\edge;\up?)$, 
for any given $n$, there is some $v\sim_\sq t$ such that $v\in \nbisim{s}{n}$. 
Suppose that this was not the case. 
Then there is some such $t$ and some $n$ such that for all $v\sim_\sq t$, $v\notin\nbisim{s}{n}$. 
From Lemma~\ref{nbisimwitness}, there is some formula, $\phi_s^n$, such that for all $w\in S$, $M_w\models\phi_s^n$ if and only if $w\in\nbisim{s}{n}^\Pi$.
Let $\psi$ be the formula $\hearts\imp\susp_\sq\phi_s^n$. Then:
\begin{enumerate}
\item $M_s\models\psi$. It is clear that $s$ is $n$-$\Pi$-bisimilar to itself, so $M_s\models\phi_s^n$.
\item $M_{t_i}\models\psi$, for $i = 1,\hdots,7$. Since $M_{t_i}\models local$, we have $M_{t_i}\models oneLabel$. We also have $M_{t_i}\models \lef\lor\righ\lor\up\lor\down$, so $M_{t_i}\not\models\hearts$. As the antecedent is not satisfied, we have $M_{t_i}\models\psi$
\item $M_t\models\psi$. This follows from the reasoning used in the previous case, since we have $M_t\models local$ and $M_t\models\up$.
\item for all $v\sim_\sq t$, if $M_v\models\hearts$, the $M_v\models\neg\psi$. Since if $M_v\models\susp_\sq\phi_s^n$ there would be some $v'\sim_\sq t$ such that $v'\in\nbisim{s}{n}$ contradicting our assumption.
\end{enumerate}
Therefore after announcing $\psi$, the worlds $s, t_1,\hdots, t_7$, and $t$ would all remain in the model $M^\psi = (S^\psi, R^\psi, V^\psi)$. 
However, for all $v\in t R_\sq\cap S^\psi$ we must have $M_v\models(\hearts\imp\phi)\et\neg\phi$, so it must be that $M_v\models\neg\hearts$.
As propositions are preserved by announcements, we have $M^\psi_t\models\know_\sq\neg\hearts$, and hence
\begin{equation}
M^\psi_s\models\susp_\sq(\righ\et(\susp_\edge(\lef\et\susp_\sq(\up\et\susp_\edge(\down\et\susp_\sq(\lef\et\susp_\edge(\righ\et\susp_\sq(\down\et\susp_\edge(\up\et\know_\sq\neg\hearts)\!)\!)\!)\!)\!)\!)\!)\!)
\end{equation}

It follows that $M_s\models\neg cyc_{ga}(\hearts)$ giving us the required contradiction. A similar argument can be given for each of the other card suits, completing the proof.
\end{enumerate}
\end{proof}

\section{Undecidability}\label{sec.together} 

In this section we show how the formulas $CB_{apa}$, $CB_{ga}$, and $CB_{ca}$ enforce the checkerboard structure of Figure~\ref{cbRep} making the satisfiability of $SAT_\Gamma$ (\ref{tiling}) equivalent to the existence of a $\Gamma$-tiling of the plane.

As the proof is complex, but uniform for the logics APAL, GAL and CAL, given the different variations of Lemmas~\ref{CB-AA},~\ref{CB-GA}~and~\ref{CB-CA}, we can give a single presentation of the remaining steps required.
To that end, we give a unified presentation of Lemma~\ref{CB-AA}, Lemma~\ref{CB-GA} and Lemma~\ref{CB-CA}.
\begin{corollary}\label{CB}
Suppose that $M = (S,\sim, V)$, $s\in S$ and either $M_s\models_{APAL} CB_{apa}\et card$, or $M_s\models_{GAL} CB_{ga}\et card$ or $M_s\models_{CAL} CB_{ca}\et card$. 
Let $U = s R(\sq;\righ?;\edge;\lef?;\sq) \cup s R(\sq;\up?;\edge;\down?;\sq) \cup s R(\sq;\lef?;\edge;\righ?;\sq)\cup s R(\sq;\down?;\edge;\up?;\sq)$.
Then, for all $s'\sim_\edge s$:
\begin{enumerate}
\item For all $n\in\N$, for all $t'\in U$, there is some $t\sim_\sq t'$ and some $v\sim_\edge s'$ such that $v\in\nbisim{t}{n}$.
\item For all $n\in\N$, for all $t \in s' R(\sq;\righ?;\edge;\lef?;\sq;\up?;\edge;\down?;\sq;\lef?;\edge;\righ?;\sq;\down?;\edge;\up?)$, there is some $v\sim_\sq t$ such that $v\in\nbisim{s'}{n}$. 
\end{enumerate}
\end{corollary}
\begin{proof}
The GAL case follows from Lemma~\ref{CB-GA}, observing that $CB_{ga}$ is known by agent $\edge$. Additionally, the APAL and CAL cases follow from the fact that $M_s\models adj$ (see (\ref{adj})) implies that if $M_s\models card$ then there is always at least one $t\in s R(\sq;\righ?;edge;\lef?;sq)$, such that $M_t\models card$ (and similarly for the other directions).
\end{proof}

The main construction is given in the following two lemmas.
\begin{lemma}\label{sat2tile}
Suppose that $M_s\models SAT_\Gamma\et\hearts$ and either $M_s\models_{APAL} CB_{apa}$, or $M_s\models_{GAL} CB_{ga}$ or $M_s\models_{CAL} CB_{ca}$.
Then $\Gamma$ can tile the plane.
\end{lemma}
\begin{proof}
Given the model $M = (S, R, V)$ and the state $s$, we construct a $\Gamma$-tiling of the plane $\N\times\N$.
To do this, we inductively define a series of maps from the points in the plane $\N\times\N$ to worlds in $S$. 
From this series we are able to build a complete tiling. Let $P_n = \{(i,j)\ |\ i+j\leq n\}$.
We show that there exists a map $\tau_n:P_n\longrightarrow S$ for each $n\in\N$ such that:
\begin{itemize}
\item $\tau_n(0,0) = s$.
\item For all $i,j\in \N$, if $i+j\leq n$, then $\tau_n(i,j)\sim_\edge s$.
\item For all $i,j\in \N$, if $n>0$ then there exists $t\in \tau_n(i,j)R(\sq;\lef?;\edge;\righ?;\sq)$ such that $\tau_n(i+1,j)\in\nbisim{t}{9(n-(i+j))}^\Pi$ 
\item For all $i,j\in \N$, if $n>0$ then there exists $t\in\tau_n(i,j)R(\sq;\up?;\edge;\down?;\sq)$ such that $\tau_n(i,j+1)\in\nbisim{t}{9(n-(i+j))}^\Pi$
\end{itemize}
We show such a construction exists for each $n$ by induction. 
The base of the induction is $\tau_n(0,0) = s$, and for the inductive step, given $(i,j)$, we suppose that $\tau_n(i',j')$ has been defined for all $(i',j')$ where $i'+j'<i+j$.
Let $m = n-(i+j)$.
There are three inductive cases to consider:
\begin{enumerate}
\item $i=0$: In this case, suppose that $\tau_n(0,j-1) = s'$. Since $s'\sim_\edge s$ we can apply the second clause of Corollary~\ref{CB} to find some $t\in s' R(\sq;\up?;\edge;\down?;\sq)$, 
and some $v\sim_\edge s'$ such that $v\in\nbisim{t}{9m}$. We set $\tau_n(0,j) = v$. Since $v\sim_\edge s'\sim_\edge s$, by transitivity we have $v\sim_\edge s$. 
Since there is some $t\in\tau_n(0,j-1)R(\sq;\up?;\edge;\down?;\sq)$ such that $\tau_n(i,j)\in\nbisim{t}{9m}$, the inductive conditions are met.
\item $j=0$: In this case, suppose that $\tau_n(i-1,0) = s'$. Since $s'\sim_\edge s$ we can apply the third clause of Corollary~\ref{CB} to find some $t\in s' R(\sq;\lef?;\edge;\righ?;\sq)$, 
and some $v\sim_\edge s'$ such that $v\in\nbisim{t}{9m}$. We set $\tau_n(i,0) = v$. Since $v\sim_\edge s'\sim_\edge s$, by transitivity we have $v\sim_\edge s$. 
Since there is some $t\in\tau_n(0,j-1)R(\sq;\lef?;\edge;\righ?;\sq)$ such that $\tau_n(i,j)\in\nbisim{t}{9m}$, the inductive conditions are met.
\item $i\neq 0$ and $j\neq 0$: In this case suppose that $\tau_n(i-1,j-1) = s'$, $\tau_n(i,j-1) = s_\up$ and $\tau_n(i-1,j) = s_\lef$. By the inductive hypothesis we have that
\begin{itemize}
\item there is some $v'\in s'R(\sq;\lef?;\edge;\righ?;\sq)$ such that $s_\lef\in\nbisim{v'}{9(m+1)}$.
\item there is some $w'\in s'R(\sq;\up?;\edge;\down?;\sq)$ such that $s_\up\in\nbisim{w'}{9(m+1)}$. 
\end{itemize}
Since there is some $v'\in s'R(\sq;\lef?;\edge;\righ?\sq)$ such that $s_\lef\in\nbisim{v'}{9(m+1)}$, 
then there must be some $t'\in s_\lef R(\sq;\righ?;\edge;\lef?;\sq)$ such that $t'\in\nbisim{s'}{9m+6}$.
We apply the last clause of Corollary~\ref{CB} to $s_\lef$, so that for all $t\in s_\lef R(\sq;\righ?;\edge;\lef?;\sq;\up?;\edge;\down?;\sq;\lef?;\edge;\righ?;\sq;\down?;\edge;\up?)$ 
such that there is some $v\sim_\sq t$ where $v\in\nbisim{s_\lef}{9m}$.
Therefore there must be some chain:
$$s_\lef\sim_\sq t_1\sim_\edge t_2\sim_\sq t_3\sim_\edge t_4\sim_\sq t_5\sim_\edge t_6\sim_\sq t_7\sim_\edge t_8$$
such that:
\begin{enumerate}
\item $M_{t_1}\models\righ$, $M_{t_2}\models\lef$ and $t'\sim_\sq t_2$. Since $t'\in s_\lef R(\sq;\righ?\edge;\lef?;\sq)$ and the conditions of $local$ (\ref{local}) are met everywhere, we can find a chain that passes through $t'$.
\item $M_{t_3}\models\up$, $M_{t_4}\models\down$ and there is some $t_\up\sim_\sq t_4$ such that $t_\up\in\nbisim{s_\up}{9m+3}$. 
Since $t'\in\nbisim{s'}{9m+6}$ and there is some $w'\in s' R(\sq;\up?;\edge;\down?;\sq)$ such that $s_\up\in\nbisim{w'}{9(m+1)}$, 
there must be some $t_\up\in t' R(\sq;\up?;\edge;\down?;\sq)$ such that $t_\up\in\nbisim{v'}{9m+3}$. Since $n$-bisimilarity is transitive, we have $t_\up\in\nbisim{s_\up}{9m+3}$.
\item $M_{t_5}\models\lef$, $M_{t_6}\models\righ$, and there is some $t'\sim_\sq t_6$ such that for some $s^*\sim_\edge s$, $t'\in\nbisim{s^*}{9m}$. 
By the third clause of Corollary~\ref{CB}, for every $n$, for every $v\in s_\up R(\sq;\lef?;\edge;\righ?)$ there is some $s^*\sim_\sq v$ such that $v'\in\nbisim{s}{n}$. 
Therefore there is some such $v$ where $v\in\nbisim{t_6}{9m+1}$, and thus there must be some $t^*\sim_\sq t_6$ where $t^*\in\nbisim{s^*}{9m}$. 
We let such a $s^*$ be $\tau_n(i,j)$.
\item $M_{t_7}\models\down$, $M_{t_8}\models\up$ and there is some $v\sim_\sq t_8$ such that $v\in\nbisim{s_\lef}{9m+3}$. Since $v\sim_\sq t_8\sim_\edge t_7\sim_\sq t^*$, 
and $t^*\in\nbisim{s^*}{9m}$, there must be some $w\in s_\lef R(\up?;\edge;\down?;\sq)$ such that $s^*\in\nbisim{w}{9m}$, as required for the induction.
\end{enumerate}
\end{enumerate}
The construction is depicted in Figure~\ref{induction}.
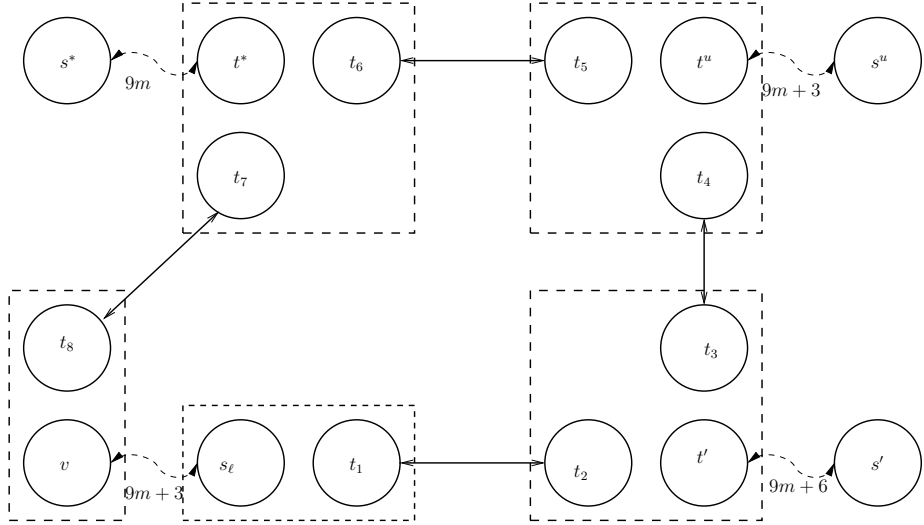
\begin{figure}
\scalebox{0.6}{
\begin{picture}(0,0)%
\includegraphics{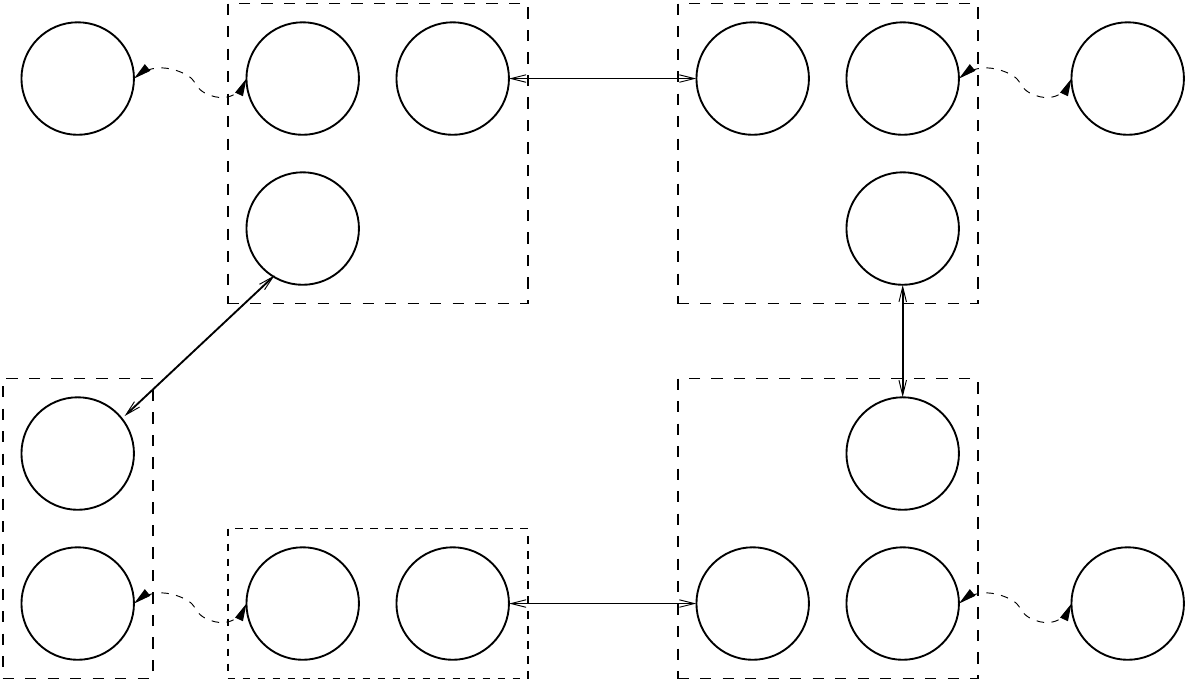}%
\end{picture}%
\setlength{\unitlength}{3947sp}%
\begingroup\makeatletter\ifx\SetFigFont\undefined%
\gdef\SetFigFont#1#2#3#4#5{%
  \reset@font\fontsize{#1}{#2pt}%
  \fontfamily{#3}\fontseries{#4}\fontshape{#5}%
  \selectfont}%
\fi\endgroup%
\begin{picture}(9487,5444)(1779,-6983)
\put(9601,-2536){\makebox(0,0)[lb]{\smash{{\SetFigFont{12}{14.4}{\rmdefault}{\mddefault}{\updefault}{\color[rgb]{0,0,0}$9m+3$}%
}}}}
\put(3976,-6436){\makebox(0,0)[lb]{\smash{{\SetFigFont{12}{14.4}{\rmdefault}{\mddefault}{\updefault}{\color[rgb]{0,0,0}$s_\lef$}%
}}}}
\put(10726,-6436){\makebox(0,0)[lb]{\smash{{\SetFigFont{12}{14.4}{\rmdefault}{\mddefault}{\updefault}{\color[rgb]{0,0,0}$s'$}%
}}}}
\put(10726,-2236){\makebox(0,0)[lb]{\smash{{\SetFigFont{12}{14.4}{\rmdefault}{\mddefault}{\updefault}{\color[rgb]{0,0,0}$s^\up$}%
}}}}
\put(2326,-2236){\makebox(0,0)[lb]{\smash{{\SetFigFont{12}{14.4}{\rmdefault}{\mddefault}{\updefault}{\color[rgb]{0,0,0}$s^*$}%
}}}}
\put(5326,-6436){\makebox(0,0)[lb]{\smash{{\SetFigFont{12}{14.4}{\rmdefault}{\mddefault}{\updefault}{\color[rgb]{0,0,0}$t_1$}%
}}}}
\put(7651,-6511){\makebox(0,0)[lb]{\smash{{\SetFigFont{12}{14.4}{\rmdefault}{\mddefault}{\updefault}{\color[rgb]{0,0,0}$t_2$}%
}}}}
\put(9001,-5236){\makebox(0,0)[lb]{\smash{{\SetFigFont{12}{14.4}{\rmdefault}{\mddefault}{\updefault}{\color[rgb]{0,0,0}$t_3$}%
}}}}
\put(8926,-3436){\makebox(0,0)[lb]{\smash{{\SetFigFont{12}{14.4}{\rmdefault}{\mddefault}{\updefault}{\color[rgb]{0,0,0}$t_4$}%
}}}}
\put(7651,-2236){\makebox(0,0)[lb]{\smash{{\SetFigFont{12}{14.4}{\rmdefault}{\mddefault}{\updefault}{\color[rgb]{0,0,0}$t_5$}%
}}}}
\put(5326,-2236){\makebox(0,0)[lb]{\smash{{\SetFigFont{12}{14.4}{\rmdefault}{\mddefault}{\updefault}{\color[rgb]{0,0,0}$t_6$}%
}}}}
\put(4126,-3436){\makebox(0,0)[lb]{\smash{{\SetFigFont{12}{14.4}{\rmdefault}{\mddefault}{\updefault}{\color[rgb]{0,0,0}$t_7$}%
}}}}
\put(2326,-5161){\makebox(0,0)[lb]{\smash{{\SetFigFont{12}{14.4}{\rmdefault}{\mddefault}{\updefault}{\color[rgb]{0,0,0}$t_8$}%
}}}}
\put(4126,-2236){\makebox(0,0)[lb]{\smash{{\SetFigFont{12}{14.4}{\rmdefault}{\mddefault}{\updefault}{\color[rgb]{0,0,0}$t^*$}%
}}}}
\put(8926,-6361){\makebox(0,0)[lb]{\smash{{\SetFigFont{12}{14.4}{\rmdefault}{\mddefault}{\updefault}{\color[rgb]{0,0,0}$t'$}%
}}}}
\put(8926,-2236){\makebox(0,0)[lb]{\smash{{\SetFigFont{12}{14.4}{\rmdefault}{\mddefault}{\updefault}{\color[rgb]{0,0,0}$t^\up$}%
}}}}
\put(2326,-6436){\makebox(0,0)[lb]{\smash{{\SetFigFont{12}{14.4}{\rmdefault}{\mddefault}{\updefault}{\color[rgb]{0,0,0}$v$}%
}}}}
\put(9676,-6661){\makebox(0,0)[lb]{\smash{{\SetFigFont{12}{14.4}{\rmdefault}{\mddefault}{\updefault}{\color[rgb]{0,0,0}$9m+6$}%
}}}}
\put(3001,-6736){\makebox(0,0)[lb]{\smash{{\SetFigFont{12}{14.4}{\rmdefault}{\mddefault}{\updefault}{\color[rgb]{0,0,0}$9m+3$}%
}}}}
\put(3001,-2461){\makebox(0,0)[lb]{\smash{{\SetFigFont{12}{14.4}{\rmdefault}{\mddefault}{\updefault}{\color[rgb]{0,0,0}$9m$}%
}}}}
\end{picture}%
}
\caption{The constructions for $\tau_n$. The curved, dashed lines indicate $n$-$\Pi$-bisimilarity, where $n$ is indicated below the line. }\label{induction}
\end{figure}
Therefore, we may assume the existence of the function $\tau_n$ for all $n$ and given the set of mappings $\tau_n$, we are now able to construct a tiling for $P^n$, for every $n$. 
For each $n\in \N$ we define a function $\lambda_n: P^n\longrightarrow\Gamma$, where for each $(i,j)\in P^n$, $\lambda_n(i,j) = \gamma$ where
$$M_{\tau_n(i,j)}\models \know_\sq(\up\imp\gamma^\up\et\down\imp\gamma^\down\et\lef\imp\gamma^\lef\et\righ\imp\gamma^\righ).$$ 
By the construction of $\tau_n$, it follows that the conditions of the tiling are met. For example, if $\lambda_n(i,j) = \gamma$, then
$M_{\tau_n(i,j)}\models\know_\sq(\up\imp\know_\edge(\gamma^\up))$. Since there is $t\in\tau_n(i,j) R(\sq;\up?;\edge;\down?;\sq)$ where $\tau_n(i,j+1)\in\nbisim{t}{1}$, we have
$M_{\tau_n(i,j+1)}\models\know_\sq(\down\imp\gamma^\up)$. This way we can show all the sides of the tiles in the partial tilings $\lambda_n$ match.

We are now able to build a complete tiling $\lambda:\N\times\N\longrightarrow\Gamma$ 
by enumerating $\N\times\N$ as $a_0,a_1,...$ where $a_0 = (0,0)$. 
We give an inductive definition of $\lambda$, whereby we suppose that $\lambda^n:\{a_0,\hdots,a_n\}\longrightarrow\Gamma$ where $\lambda^0(0,0)= \gamma$ such that  
$$M_s\models \know_\sq(\up\imp\gamma^\up\et\down\imp\gamma^\down\et\lef\imp\gamma^\lef\et\righ\imp\gamma^\righ).$$ 
Our induction is built on the hypothesis that for each $n$ there is an infinite subset $X^n\subseteq\N$ such that for all $x\in X$, for all $i\leq n$, $\lambda_x(a_i) = \lambda^n(a_i)$.
This is clearly true for $n = 0$, where $X^0 = \N$. Given the hypothesis holds for $n-1$, we set $\lambda^{n}(a_i) = \lambda^{n-1}(a_i)$ for all $i<n$. 
Now for each $x\in X^{n-1}$, $\lambda_x(a_n)$ may be a tile of $\Gamma$, or may be undefined. However, if $a_n = (i,j)$, $\lambda_x(a_n)$ may only be undefined if $x<i+j$. 
As there can only be a finite number of such cases, and $X^{n-1}$ is infinite, it follows that, $\lambda_x(a_n)\in\Gamma$ for infinitely many $x$. 
As $\Gamma$ is a finite set, there must be some $\gamma\in\Gamma$ such that $\lambda_x(a_n) = \gamma$ for all $x\in X'$ where $X'$ is an infinite subset of $X^{n-1}$.
We let $\lambda^n(a_n) = \gamma$ and $X^n = X'$, and the induction hypothesis holds.

As $\lambda^n$ is a monotonically increasing function it has a well-defined limit: $\lambda(i,j) = \gamma$ where for some $x\in \N$, $\lambda^x(i,j) = \gamma$.
As the sides of the tiles all match, this is sufficient to show that $\Gamma$ can tile the plane. 
\end{proof}

For the converse direction we are able to give a more direct proof.
\begin{lemma}\label{tile2sat}
Suppose that $\Gamma$ can tile the plane. Then there is some model $M = (S, R, V)$ and some state $s\in S$ such that  
\begin{enumerate}
\item $M_s\models_{APAL} SAT_\Gamma\et CB_{apa}\et\hearts$
\item $M_s\models_{GAL} SAT_\Gamma\et CB_{ga}\et\hearts$
\item $M_s\models_{CAL} SAT_\Gamma\et CB_{ca}\et\hearts$
\end{enumerate}
\end{lemma}

\begin{proof}
Suppose that $\Gamma$ can tile the plane via the function $\lambda:\N\times\N\longrightarrow\Gamma$.
We build the model $M = (S, R, V)$ where:
\begin{itemize}
\item $S = \N\times\N\times\{\up,\down,\lef,\righ, mid\}$, so the state $(i,j,k)$ represents the $k$ side of the tile at the square $(i,j)$, where $k=mid$ represents the center of the tile.
\item $R_\sq = \{((i,j,k),(i',j',k'))\ |\ i = i'\ \text{and}\ j = j'\}$
\item $(i,j,\up)R_\edge = \{(i,j,\up),(i,j+1,\down)\}$
\item $(i,j,\down)R_\edge = \{(i,j,\down),(i,j-1,\up)\}$ or $\{(i,j,\down)\}$ if $j = 0$.
\item $(i,j,\lef)R_\edge = \{(i,j,\lef),(i-1,j,\righ)\}$ of $\{(i,j,\lef)\}$ if $i = 0$.
\item $(i,j,\righ)R_\edge = \{(i,j,\righ),(i+1,j,\lef)\}$
\item $(i,j,mid)R_\edge = \{(i',j',mid)\ | (i',j')\in \N\times\N\}$
\item for all colours $c\in C$, $V(c) = \{(i,j,k)\ |\ \lambda(i,j)^k = c\}$ (We assume that $\lambda(i,j)^{mid}$ is always a fixed colour, say {\it white}.)
\item for all directions $k\in\{\up,\down,\lef,\righ\}$, $V(k) = \{(i,j,k')\ |\ k' = k\}$.
\item $V(\hearts) = \{(i,j,mid)\ |\ i\textrm{ is even, and }j\textrm{ is even}\}$
\item $V(\clubs) = \{(i,j,mid)\ |\ i\textrm{ is odd, and }j\textrm{ is even}\}$
\item $V(\diamonds) = \{(i,j,mid)\ |\ i\textrm{ is odd, and }j\textrm{ is odd}\}$
\item $V(\spades) = \{(i,j,mid)\ |\ i\textrm{ is even, and }j\textrm{ is odd}\}$
\end{itemize}
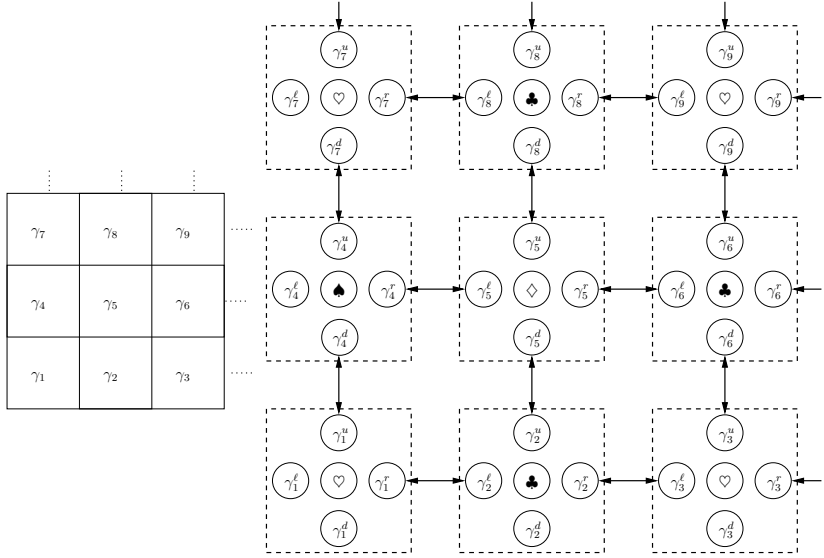
\begin{figure}
\scalebox{0.5}{
\begin{picture}(0,0)%
\includegraphics{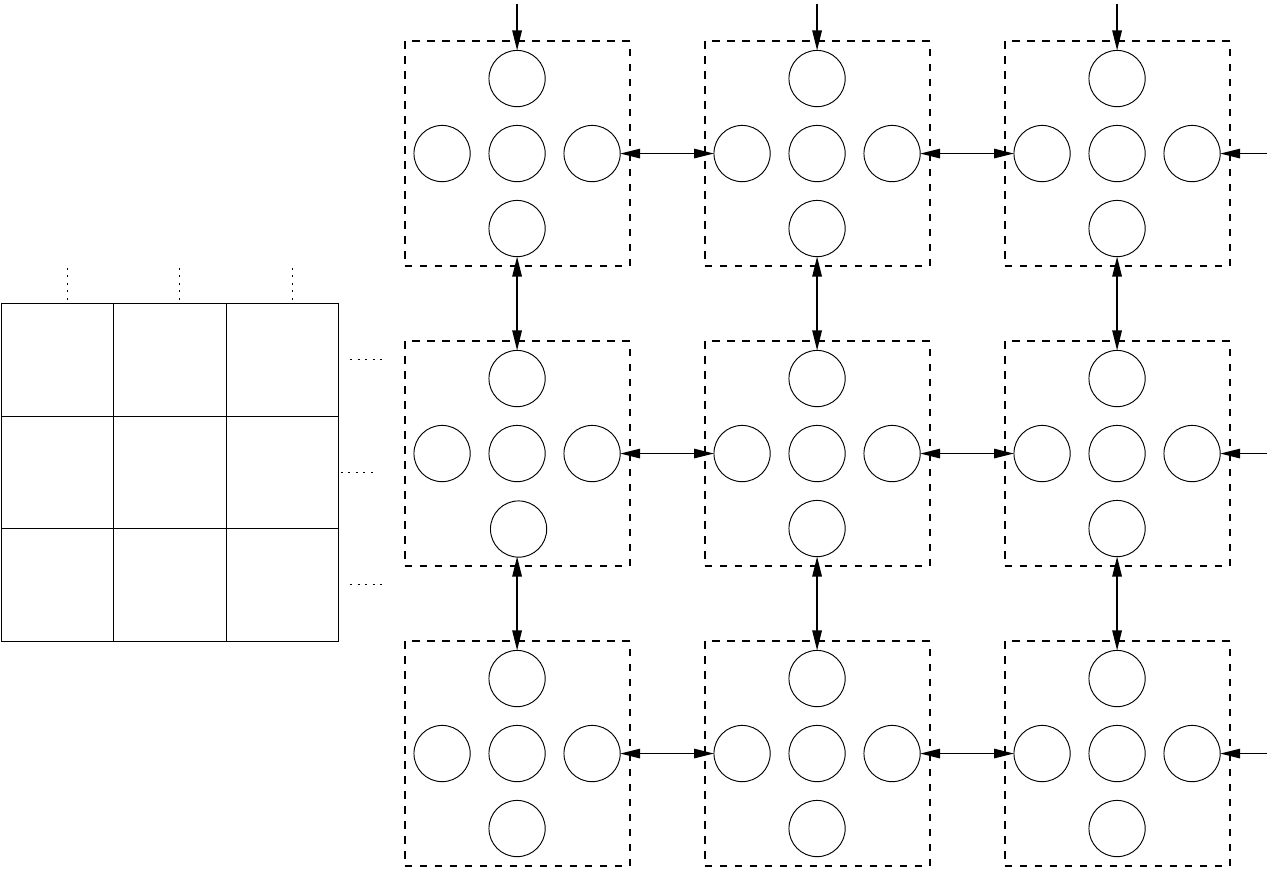}%
\end{picture}%
\setlength{\unitlength}{3947sp}%
\begingroup\makeatletter\ifx\SetFigFont\undefined%
\gdef\SetFigFont#1#2#3#4#5{%
  \reset@font\fontsize{#1}{#2pt}%
  \fontfamily{#3}\fontseries{#4}\fontshape{#5}%
  \selectfont}%
\fi\endgroup%
\begin{picture}(10159,6944)(-2036,-8183)
\put(4426,-6736){\makebox(0,0)[lb]{\smash{{\SetFigFont{12}{14.4}{\rmdefault}{\mddefault}{\updefault}{\color[rgb]{0,0,0}$\gamma_2^\up$}%
}}}}
\put(-1724,-4186){\makebox(0,0)[lb]{\smash{{\SetFigFont{12}{14.4}{\rmdefault}{\mddefault}{\updefault}{\color[rgb]{0,0,0}$\gamma_7$}%
}}}}
\put(-1724,-5086){\makebox(0,0)[lb]{\smash{{\SetFigFont{12}{14.4}{\rmdefault}{\mddefault}{\updefault}{\color[rgb]{0,0,0}$\gamma_4$}%
}}}}
\put(-1724,-5986){\makebox(0,0)[lb]{\smash{{\SetFigFont{12}{14.4}{\rmdefault}{\mddefault}{\updefault}{\color[rgb]{0,0,0}$\gamma_1$}%
}}}}
\put(-824,-5986){\makebox(0,0)[lb]{\smash{{\SetFigFont{12}{14.4}{\rmdefault}{\mddefault}{\updefault}{\color[rgb]{0,0,0}$\gamma_2$}%
}}}}
\put( 76,-5986){\makebox(0,0)[lb]{\smash{{\SetFigFont{12}{14.4}{\rmdefault}{\mddefault}{\updefault}{\color[rgb]{0,0,0}$\gamma_3$}%
}}}}
\put(-824,-5086){\makebox(0,0)[lb]{\smash{{\SetFigFont{12}{14.4}{\rmdefault}{\mddefault}{\updefault}{\color[rgb]{0,0,0}$\gamma_5$}%
}}}}
\put( 76,-5086){\makebox(0,0)[lb]{\smash{{\SetFigFont{12}{14.4}{\rmdefault}{\mddefault}{\updefault}{\color[rgb]{0,0,0}$\gamma_6$}%
}}}}
\put(-824,-4186){\makebox(0,0)[lb]{\smash{{\SetFigFont{12}{14.4}{\rmdefault}{\mddefault}{\updefault}{\color[rgb]{0,0,0}$\gamma_8$}%
}}}}
\put( 76,-4186){\makebox(0,0)[lb]{\smash{{\SetFigFont{12}{14.4}{\rmdefault}{\mddefault}{\updefault}{\color[rgb]{0,0,0}$\gamma_9$}%
}}}}
\put(4426,-7336){\makebox(0,0)[lb]{\smash{{\SetFigFont{12}{14.4}{\rmdefault}{\mddefault}{\updefault}{\color[rgb]{0,0,0}$\clubs$}%
}}}}
\put(4426,-4936){\makebox(0,0)[lb]{\smash{{\SetFigFont{12}{14.4}{\rmdefault}{\mddefault}{\updefault}{\color[rgb]{0,0,0}$\diamonds$}%
}}}}
\put(2026,-4936){\makebox(0,0)[lb]{\smash{{\SetFigFont{12}{14.4}{\rmdefault}{\mddefault}{\updefault}{\color[rgb]{0,0,0}$\spades$}%
}}}}
\put(6826,-7336){\makebox(0,0)[lb]{\smash{{\SetFigFont{12}{14.4}{\rmdefault}{\mddefault}{\updefault}{\color[rgb]{0,0,0}$\hearts$}%
}}}}
\put(6826,-2536){\makebox(0,0)[lb]{\smash{{\SetFigFont{12}{14.4}{\rmdefault}{\mddefault}{\updefault}{\color[rgb]{0,0,0}$\hearts$}%
}}}}
\put(2026,-7936){\makebox(0,0)[lb]{\smash{{\SetFigFont{12}{14.4}{\rmdefault}{\mddefault}{\updefault}{\color[rgb]{0,0,0}$\gamma_1^\down$}%
}}}}
\put(1426,-7336){\makebox(0,0)[lb]{\smash{{\SetFigFont{12}{14.4}{\rmdefault}{\mddefault}{\updefault}{\color[rgb]{0,0,0}$\gamma_1^\lef$}%
}}}}
\put(2026,-6736){\makebox(0,0)[lb]{\smash{{\SetFigFont{12}{14.4}{\rmdefault}{\mddefault}{\updefault}{\color[rgb]{0,0,0}$\gamma_1^\up$}%
}}}}
\put(2026,-7336){\makebox(0,0)[lb]{\smash{{\SetFigFont{12}{14.4}{\rmdefault}{\mddefault}{\updefault}{\color[rgb]{0,0,0}$\hearts$}%
}}}}
\put(2551,-7336){\makebox(0,0)[lb]{\smash{{\SetFigFont{12}{14.4}{\rmdefault}{\mddefault}{\updefault}{\color[rgb]{0,0,0}$\gamma_1^\righ$}%
}}}}
\put(3826,-7336){\makebox(0,0)[lb]{\smash{{\SetFigFont{12}{14.4}{\rmdefault}{\mddefault}{\updefault}{\color[rgb]{0,0,0}$\gamma_2^\lef$}%
}}}}
\put(4426,-7936){\makebox(0,0)[lb]{\smash{{\SetFigFont{12}{14.4}{\rmdefault}{\mddefault}{\updefault}{\color[rgb]{0,0,0}$\gamma_2^\down$}%
}}}}
\put(6826,-7936){\makebox(0,0)[lb]{\smash{{\SetFigFont{12}{14.4}{\rmdefault}{\mddefault}{\updefault}{\color[rgb]{0,0,0}$\gamma_3^\down$}%
}}}}
\put(6826,-6736){\makebox(0,0)[lb]{\smash{{\SetFigFont{12}{14.4}{\rmdefault}{\mddefault}{\updefault}{\color[rgb]{0,0,0}$\gamma_3^\up$}%
}}}}
\put(7426,-7336){\makebox(0,0)[lb]{\smash{{\SetFigFont{12}{14.4}{\rmdefault}{\mddefault}{\updefault}{\color[rgb]{0,0,0}$\gamma_3^\righ$}%
}}}}
\put(2026,-5536){\makebox(0,0)[lb]{\smash{{\SetFigFont{12}{14.4}{\rmdefault}{\mddefault}{\updefault}{\color[rgb]{0,0,0}$\gamma_4^\down$}%
}}}}
\put(1426,-4936){\makebox(0,0)[lb]{\smash{{\SetFigFont{12}{14.4}{\rmdefault}{\mddefault}{\updefault}{\color[rgb]{0,0,0}$\gamma_4^\lef$}%
}}}}
\put(2026,-4336){\makebox(0,0)[lb]{\smash{{\SetFigFont{12}{14.4}{\rmdefault}{\mddefault}{\updefault}{\color[rgb]{0,0,0}$\gamma_4^\up$}%
}}}}
\put(4426,-4336){\makebox(0,0)[lb]{\smash{{\SetFigFont{12}{14.4}{\rmdefault}{\mddefault}{\updefault}{\color[rgb]{0,0,0}$\gamma_5^\up$}%
}}}}
\put(4426,-5536){\makebox(0,0)[lb]{\smash{{\SetFigFont{12}{14.4}{\rmdefault}{\mddefault}{\updefault}{\color[rgb]{0,0,0}$\gamma_5^\down$}%
}}}}
\put(5026,-4936){\makebox(0,0)[lb]{\smash{{\SetFigFont{12}{14.4}{\rmdefault}{\mddefault}{\updefault}{\color[rgb]{0,0,0}$\gamma_5^\righ$}%
}}}}
\put(2626,-4936){\makebox(0,0)[lb]{\smash{{\SetFigFont{12}{14.4}{\rmdefault}{\mddefault}{\updefault}{\color[rgb]{0,0,0}$\gamma_4^\righ$}%
}}}}
\put(3826,-4936){\makebox(0,0)[lb]{\smash{{\SetFigFont{12}{14.4}{\rmdefault}{\mddefault}{\updefault}{\color[rgb]{0,0,0}$\gamma_5^\lef$}%
}}}}
\put(4426,-1936){\makebox(0,0)[lb]{\smash{{\SetFigFont{12}{14.4}{\rmdefault}{\mddefault}{\updefault}{\color[rgb]{0,0,0}$\gamma_8^\up$}%
}}}}
\put(1951,-3136){\makebox(0,0)[lb]{\smash{{\SetFigFont{12}{14.4}{\rmdefault}{\mddefault}{\updefault}{\color[rgb]{0,0,0}$\gamma_7^\down$}%
}}}}
\put(2026,-1936){\makebox(0,0)[lb]{\smash{{\SetFigFont{12}{14.4}{\rmdefault}{\mddefault}{\updefault}{\color[rgb]{0,0,0}$\gamma_7^\up$}%
}}}}
\put(4426,-3136){\makebox(0,0)[lb]{\smash{{\SetFigFont{12}{14.4}{\rmdefault}{\mddefault}{\updefault}{\color[rgb]{0,0,0}$\gamma_8^\down$}%
}}}}
\put(4951,-2536){\makebox(0,0)[lb]{\smash{{\SetFigFont{12}{14.4}{\rmdefault}{\mddefault}{\updefault}{\color[rgb]{0,0,0}$\gamma_8^\righ$}%
}}}}
\put(6826,-4336){\makebox(0,0)[lb]{\smash{{\SetFigFont{12}{14.4}{\rmdefault}{\mddefault}{\updefault}{\color[rgb]{0,0,0}$\gamma_6^\up$}%
}}}}
\put(6826,-5536){\makebox(0,0)[lb]{\smash{{\SetFigFont{12}{14.4}{\rmdefault}{\mddefault}{\updefault}{\color[rgb]{0,0,0}$\gamma_6^\down$}%
}}}}
\put(7426,-4936){\makebox(0,0)[lb]{\smash{{\SetFigFont{12}{14.4}{\rmdefault}{\mddefault}{\updefault}{\color[rgb]{0,0,0}$\gamma_6^\righ$}%
}}}}
\put(6826,-3136){\makebox(0,0)[lb]{\smash{{\SetFigFont{12}{14.4}{\rmdefault}{\mddefault}{\updefault}{\color[rgb]{0,0,0}$\gamma_9^\down$}%
}}}}
\put(6826,-1936){\makebox(0,0)[lb]{\smash{{\SetFigFont{12}{14.4}{\rmdefault}{\mddefault}{\updefault}{\color[rgb]{0,0,0}$\gamma_9^\up$}%
}}}}
\put(6226,-2536){\makebox(0,0)[lb]{\smash{{\SetFigFont{12}{14.4}{\rmdefault}{\mddefault}{\updefault}{\color[rgb]{0,0,0}$\gamma_9^\lef$}%
}}}}
\put(7426,-2536){\makebox(0,0)[lb]{\smash{{\SetFigFont{12}{14.4}{\rmdefault}{\mddefault}{\updefault}{\color[rgb]{0,0,0}$\gamma_9^\righ$}%
}}}}
\put(2551,-2536){\makebox(0,0)[lb]{\smash{{\SetFigFont{12}{14.4}{\rmdefault}{\mddefault}{\updefault}{\color[rgb]{0,0,0}$\gamma_7^\righ$}%
}}}}
\put(2026,-2536){\makebox(0,0)[lb]{\smash{{\SetFigFont{12}{14.4}{\rmdefault}{\mddefault}{\updefault}{\color[rgb]{0,0,0}$\hearts$}%
}}}}
\put(1426,-2536){\makebox(0,0)[lb]{\smash{{\SetFigFont{12}{14.4}{\rmdefault}{\mddefault}{\updefault}{\color[rgb]{0,0,0}$\gamma_7^\lef$}%
}}}}
\put(5026,-7336){\makebox(0,0)[lb]{\smash{{\SetFigFont{12}{14.4}{\rmdefault}{\mddefault}{\updefault}{\color[rgb]{0,0,0}$\gamma_2^\righ$}%
}}}}
\put(6226,-7336){\makebox(0,0)[lb]{\smash{{\SetFigFont{12}{14.4}{\rmdefault}{\mddefault}{\updefault}{\color[rgb]{0,0,0}$\gamma_3^\lef$}%
}}}}
\put(6226,-4936){\makebox(0,0)[lb]{\smash{{\SetFigFont{12}{14.4}{\rmdefault}{\mddefault}{\updefault}{\color[rgb]{0,0,0}$\gamma_6^\lef$}%
}}}}
\put(4426,-2536){\makebox(0,0)[lb]{\smash{{\SetFigFont{12}{14.4}{\rmdefault}{\mddefault}{\updefault}{\color[rgb]{0,0,0}$\clubs$}%
}}}}
\put(6826,-4936){\makebox(0,0)[lb]{\smash{{\SetFigFont{12}{14.4}{\rmdefault}{\mddefault}{\updefault}{\color[rgb]{0,0,0}$\clubs$}%
}}}}
\put(3826,-2536){\makebox(0,0)[lb]{\smash{{\SetFigFont{12}{14.4}{\rmdefault}{\mddefault}{\updefault}{\color[rgb]{0,0,0}$\gamma_8^\lef$}%
}}}}
\end{picture}%
}
\caption{A representation of the construction used for Lemma~\ref{tile2sat}. The edge relations between the centers of all squares have been omitted for clarity.}\label{cb2tile}
\end{figure}

This construction is presented in Figure~\ref{cb2tile}, and realizes the construction presented in Figure~\ref{cbRep}.
We can now confirm that the constructed model satisfies the necessary formulas. Let $s = (0,0,m)$. Then
\begin{enumerate}
\item $M_s\models oneCol$ (\ref{oneCol}). Since the $edge$-relation links the centers of all squares, and every world belongs to a square, this requires every world $t\in S$ to satisfy a unique colour atom $c\in C$. This is guaranteed by the definition of $V(c)$.
\item $M_s\models tile_\Gamma$ (\ref{tile-gamma}). Since the $\edge$-relation links the centers of all squares this requires that at the center of every square, the formula (\ref{tile-gamma}) requires that for every square, there is some tile $\gamma$ such that any state labelled by the direction $k$ in that square satisfies the colour proposition $\gamma^k$. As the square relations relates worlds $(i,j,k)$ and $(i',j',k')$ where $i = i'$ and $j=j'$, worlds in the same square correspond to to the same tile $\lambda(i,j)$ and are coloured accordingly.
\item $M_s\models match$ (\ref{match}). The $\edge$-relation connects the centers of all squares, so (\ref{match}) requires that at every world in every square, the edge agent knows the colour. Let the given world be $t= (i,j,k)$. There are five cases to consider:
  \begin{enumerate}
  \item If the world is a center world (i.e. $k = mid$) then it must be coloured white. As all center worlds are connected by $\edge$ only to center worlds, at a center world, the edge agent knows the colour is white.
  \item If the world is a right world ($k = \righ$), then its colour is $\lambda(i,j)^\righ$. As $t R_\edge=\{(i,j,\righ),(i+1,j,\lef)\}$ we have for all $t'\in t R_\edge$, the colour of $t'$ is either $\lambda(i,j)^\righ$ or $\lambda(i+1,j,\lef)$. From Definition~\ref{def:tiling} it follows that $\lambda(i,j)^\righ = \lambda(i+1,j)^\lef$, so the colour is known to agent $\edge$.
  \item Similar arguments can be given where $k = \lef,\ \up,$ or $\down$. In the case $k = \lef$ and $i = 0$, or $k = \down$ and $j =0$, we note that $t R_\edge$ is a singleton so the result follows trivially.
  \end{enumerate}
\item $M_s\models SAT_\Gamma$ (\ref{tiling}). This follows from the previous three cases. 
\item $M_s\models \know_\edge\know_\sq local$ (\ref{local}). As $local$ is a purely epistemic formula, it is straightforward to check that all the required properties are satisfied by $M_s$.
\item $M_s\models \know_\edge\know_\sq cyc_{apa}$ (\ref{CBA}). 
Suppose $t = (i,j,mid)\in S$ and $M_t\models\hearts$. After every possible announcement, $\psi$, that can be made there are two possibilities. If 
$t_0 = t$, $t_1=(i,j,\righ)$, $t_2 = (i+1,j,\lef)$, $t_3=(i+1,j,\up)$, $t_4 = (i+1,j+1,\down)$, $t_5= (i+1,j+1,\lef)$, $t_6= (i,j+1,\righ)$, $t_7=(i,j+1,\down)$, and $t_8=(i,j,\up)$ all remain in $S^\psi$, in which case $\{t_1\} = V(\righ)\cap t_0 R_\sq$, $\{t_2\} = V(\lef)\cap t_1 R_\edge$, $\{t_3\} = V(\up)\cap t_2 R_\sq$, $\{t_4\} = V(\down)\cap t_3 R_\edge$, $\{t_5\} = V(\lef)\cap t_4 R_\sq$, $\{t_6\} = V(\righ)\cap t_5 R_\edge$, $\{t_7\} = V(\down)\cap t_6 R_\sq$,  $\{t_8\} = V(\up)\cap t_7 R_\edge$ and $t\in t_8 R_\sq$. 
Alternatively one of the worlds $t_0,\hdots,t_8$ does not remain in $S^\psi$, so one of $V(\righ)\cap t_0 R^\psi_\sq$, $V(\lef)\cap t_1 R^\psi_\edge$, $V(\up)\cap t_2 R^\psi_\sq$, $V(\down)\cap t_3 R^\psi_\edge$, $V(\lef)\cap t_4 R^\psi_\sq$, $V(\righ)\cap t_5 R^\psi_\edge$, $V(\down)\cap t_6 R^\psi_\sq$,  $V(\up)\cap t_7 R^\psi_\edge$ or $t_8 R^\psi_\sq$ is empty. In either case we have
$$\begin{array}{l}
M^\psi_t\models\know_\sq(\righ\imp\know_\edge(\lef\imp\know_\sq(\up\imp\know_\edge(\down\imp\know_\sq(\lef\imp\qquad\quad\\
              \hfill\know_\edge(\righ\imp\know_\sq(\down\imp\know_\edge(\up\imp\susp_\sq \hearts)))))))).
\end{array}
$$
Similar arguments can be given for the other suits, so it follows that for all $t\in S$, $M_t\models cyc_{apa}$.
\item $M_t\models ck_{apa}$ (\ref{CBA}). 
Suppose $t = (i,j,mid)\in S$ and $M_t\models\hearts$. Suppose that there is some announcement $\psi$ such that for $t'\in V(\clubs)$, $t'\notin t R^\psi_\edge$. 
As $v = (i+1,j,mid)\sim_\edge t$, and $v\in V(\clubs)$, we must have $v\notin S^\psi$. Therefore if $t_0 = t$, $t_1=(i,j,\righ)$, 
and $t_2 = (i+1,j,\lef)$ remain in $S^\psi$ then, by the same reasoning as in the previous case, 
we will have $M^\psi_t\models\know_\sq(\righ\imp\know_\edge(\lef\imp\know_\sq\neg\clubs))$. 
If any of the worlds $t_0,\ t_1$, or $t_2$ do not remain in $S^\psi$ then the formula is vacuously true. 
Similar arguments can be given for the other suits, so for all $t\in S$, $M_t\models ck_{apa}$.  
\item $M_s\models \know_\edge\know_\sq cyc_{ga}$ (\ref{CBG}). 
Suppose that $M_t\models\hearts$. Suppose that $\psi$ is some announcement that is known to agent $\square$ at $t$.  
Let $t_0 = t$, $t_1=(i,j,\righ)$, $t_2 = (i+1,j,\lef)$, $t_3=(i+1,j,\up)$, $t_4 = (i+1,j+1,\down)$, $t_5= (i+1,j+1,\lef)$, $t_6= (i,j+1,\righ)$, $t_7=(i,j+1,\down)$, and $t_8=(i,j,\up)$.
As agent-$\sq$ knows $\psi$, but cannot distinguish $t$ from $t_1$ and $t_8$, it follows that $t_0,\ t_1$ and $t_8$ remain in $S^\psi$. Regardless of whether the other states also remain in $S^\psi$, we have $M_t\models cyc_{ga}$ in much that same way that $M_t\models cyc_{apa}$. 
\item $M_s\models \know_\edge\know_\sq ck_{ga}$ (\ref{CBG}). Suppose that $t = (i,j,mid)$ and $M_t\models\hearts$. As agent $\edge$ cannot distinguish any of the worlds $(i',j',mid)$ from $t$, any announcement $\psi$ that is known to agent $\edge$ must preserve all of those worlds. Therefore after making such an announcement, regardless of whether $(i,j,\righ)$ and $(i+1,j,\lef)$ remain, we have $M^\psi_t\models\know_\sq(\lef\imp\know_\edge(\righ\imp\susp_\sq\clubs$. Similar arguments can be given for other directions and suits so we have $M_t\models ck_{ga}$.
\item $M_s\models \know_\edge\know_\sq cyc_{ca}$ (\ref{CBC}).
Suppose that $M_t\models\hearts$. Suppose that $\psi$ is a conjunction of announcements that are known to agents $\square$ or $\edge$ at $t$.  
Let $t_0 = t$, $t_1=(i,j,\righ)$, $t_2 = (i+1,j,\lef)$, $t_3=(i+1,j,\up)$, $t_4 = (i+1,j+1,\down)$, $t_5= (i+1,j+1,\lef)$, $t_6= (i,j+1,\righ)$, $t_7=(i,j+1,\down)$, and $t_8=(i,j,\up)$.
As the announcement of $\psi$ can only possibly remove these worlds, it follows that $M_t\models cyc_{ca}$, in the same way that $M_t\models cyc_{apa}$. 
\item $M_s\models \know_\edge\know_\sq ck_{ca}$ (\ref{CBC}). Suppose that $t = (i,j,mid)$ and $M_t\models\hearts$. Let $\psi$ be an conjunction of announcements that are known to either agent $\sq$ or agent $\edge$. If $S^\psi\cap V(\clubs)$ is empty, then $t' = (i+1,j,mid)\notin S^\psi$. Therefore, $t R^\psi(\sq;\righ?;\edge;\lef;\sq;\clubs?)$ is empty, so $M^\psi_t\models\know_\sq(\righ\imp\know_\edge(\lef\imp\know_\sq\neg\clubs))$. Similar arguments can be given for other directions and suits so the result follows.
\end{enumerate}
As we have show that $M_s$ satisfies all the necessary sub-formulas, the result follows.
\end{proof}

\begin{theorem}
The satisfiability problem for  APAL, GAL and CAL is undecidable, provided that there is more than one agent in the system. In the case of a single agent, all logics are decidable. 
\end{theorem} 
\begin{proof}
Lemmas~\ref{sat2tile}~and~\ref{tile2sat} show that the satisfiability problems for APAL,  GAL, and CAL are equivalent to an undecidable tiling problem and therefore are not recursively enumerable. 
The single agent case is given as Proposition 18 of \cite{balbianietal:2008} (the arbitrary announcement operator is then definable, so that the logic is equally expressive as the base epistemic logic), and the single agent fragments of GAL and CAL can be shown similarly.
\end{proof}

\section{Future work} \label{sec.future}
The previous sections give an essentially negative result: we are unable to determine the satisfiability of formulas in APAL, GAL, and CAL.
However, the core purposes of these logics are to determine what agents can achieve by sharing knowledge, and this is still a question of great interest.
The undecidability proof given here does not use an intrinsic property of knowledge sharing, but rather exploits the power of quantifying over language.
With this in mind, we consider the question of what other logics could successfully reason about knowledge sharing agents.
We list some of the possible answers below:
\begin{itemize}
\item In the logics presented in this paper we assumed that accessibility relations are equivalence relations, and correspondingly we quantify over {\em knowledge} formulas, not over {\em beliefs}. Related to this, the quantifications are over {\em truthful} announcements, and not over announcements that are merely believed to be true, but that may be mistaken or even intentionally false (as in lies). The logics presented in this paper can be straightforwardly generalized to such scenarios, by the dual strategy of $(i)$ considering arbitrary accessibility relations (or, for example KD45 accessibility relations, for modelling consistent beliefs), and $(ii)$ not using the standard model restriction semantics of public announcements, but the `arrow restriction'  (accessibility relation restricting) semantics for public announcements of \cite{gerbrandy:1999}.
\item One approach could be to consider generalizations of announcements as the medium of informative updates.
By replacing the public announcements in arbitrary public announcement logic with refinements \cite{hvdetal.felax:2010,bozzellietal.inf:2014}
it is possible to define a decidable logic that may be used to reason about the quantification of all informative updates (and not merely public updates: announcements). This logic is called refinement modal logic (RML).
The essential difference between  APAL  and RML is, that while  APAL  has an operator that quantifies over all public announcements, RML has an operator that quantifies over all refinements.  Recent work \cite{hales2013arbitrary} presents a logic that quantifies over event models \cite{baltagetal:1998}, and it is shown there that this so-called arbitrary action model logic is equally expressive as refinement modal logic. One can in fact synthesize an action model from the formula supposedly true after a refinement. Hales' approach can clearly be adapted to quantify over group event models, and this would be a valuable generalization of group announcement logic. Like RML, such a logic would then be decidable.
\item 
A way towards a decidable arbitrary public announcement logic may be to only allow quantification over positive formulas. 
In a {\em positive formula} all $\know_\agent$ operators are in the scope of an even number of negations. 
The positive formulas are also known as the {\em universal fragment}.
This means agents may announce what they know, and what they know other agents know, and so on.
However, they may not announce what they don't know, or what they know other agents don't know.
The appealing thing about these announcements is that once they are made, they will always remain true, so they have a monotonic nature.
This suggests that a logic of arbitrary (or group) positive announcements may be more computationally amenable than GAL. However, there may well be various other roads leading to a decidable arbitrary public announcement logic. For example,  only to quantify over formulas that have a restricted epistemic (modal) depth, or variations of first-order epistemic logics with explicit quantifiers as in \cite{belardinellietal:2015}. Yet another way to decidability is to consider arbitrary announcements on specific domains, as in the recent \cite{charrieretal.aamas:2015}.
\end{itemize}     

\bibliographystyle{abbrv}

\AdditionalAuthorAddressEmail{Thomas {\AA}gotnes}{University of Bergen}{thomas.agotnes@infomedia.uib.no}
\AdditionalAuthorAddressEmail{Hans van Ditmarsch}{LORIA, CNRS - Universit\'{e} de Lorraine}{hans.van-ditmarsch@loria.fr}
\AdditionalAuthorAddressEmail{Tim French}{The University of Western Australia}{tim.french@uwa.edu.au}

\end{document}